\documentclass[final,leqno,onefignum]{siamltexmm}
\usepackage{verbatim}
\usepackage{mathrsfs}
\usepackage{amsmath}
\usepackage{amssymb}
\usepackage{esint}

\makeatletter
\let\siamproof\proof
\renewcommand{\proof}{\@ifnextchar[{\myproof}{\siamproof}}
\def\myproof[#1]{\par{\color{black}{\mbox{\color{header1}{\it #1}.$\,$}\color{black}}} \ignorespaces\color{black}}

  \newtheorem{lem}{\protect\lemmaname}
  \newtheorem{thm}{\protect\theoremname}
  \newtheorem{cor}{\protect\corollaryname}
  \newtheorem{rem}{\protect\remarkname}
  \newtheorem{defn}{\protect\definitionname}
  \newtheorem{prop}{\protect\propositionname}

\usepackage{import}
\usepackage{subcaption}
\makeatother

  \providecommand{\definitionname}{Definition}
  \providecommand{\lemmaname}{Lemma}
  \providecommand{\propositionname}{Proposition}
  \providecommand{\remarkname}{Remark}
  \providecommand{\corollaryname}{Corollary}
  \providecommand{\theoremname}{Theorem}

\newcommand{\slugmaster}{%
\slugger{siims}{2014}{xx}{x}{x--x}}

\title{Lifting for Blind Deconvolution in Random Mask Imaging: Identifiability
and Convex Relaxation\thanks{ This work was supported by ONR grant N00014-11-1-0459, and NSF grants CCF-1415498 and CCF-1422540.}}

\author{Sohail~Bahmani\thanks{School of Electrical and Computer Engineering, Georgia Institute of Technology, Atlanta, GA, 30332-0250 USA. (\email{sohail.bahmani@ece.gatech.edu, jrom@ece.gatech.edu})} \and Justin~Romberg\footnotemark[2]}

\begin{document}
\global\long\def\mb#1{\boldsymbol{#1}}
\global\long\def\mbb#1{\mathbb{#1}}
\global\long\def\mc#1{\mathcal{#1}}
\global\long\def\mcc#1{\mathscr{#1}}
\global\long\def\mr#1{\mathrm{#1}}
\global\long\def\msf#1{\mathsf{#1}}

\maketitle

\begin{abstract}
In this paper we analyze the blind deconvolution of an image and
an unknown blur in a coded imaging system.  The measurements
consist of subsampled convolution of an unknown blurring kernel
with multiple random binary modulations (coded masks) of the image.
To perform the deconvolution, we consider a standard lifting  of the
image and the blurring kernel that transforms the measurements into
a set of linear equations of the matrix formed by their outer product.
Any rank-one solution to this system of equation provides a valid pair
of an image and a blur.

We first express the necessary and sufficient conditions for the
uniqueness of a rank-one solution under some additional assumptions
(uniform subsampling and no limit on the number of coded masks).
These conditions are special case of a previously established result
regarding identifiability in the matrix completion problem.  We also
characterize a low-dimensional subspace model for the blur kernel
that is sufficient to guarantee identifiability, including the
interesting instance of ``bandpass'' blur kernels.

Next, assuming the bandpass model for the blur kernel, we show that the
image and the blur kernel can be found using nuclear norm minimization.
Our main results show that recovery is achieved (with high probability)
when the number of masks is on the order of $\mu\log^{2}L\,\log\frac{Le}{\mu}\,\log\log\left(N+1\right)$
where $\mu$ is the \emph{coherence} of the blur, $L$ is the dimension
of the image, and $N$ is the number of measured samples per mask.
\end{abstract}
\begin{keywords}
blind deconvolution, coded mask imaging, lifting, nuclear norm minimization
\end{keywords}
\section{\label{sec:Introduction}Introduction}

The blind deconvolution problem has been encountered in many fields
including astronomical, microscopic, and medical imaging, computational
photography, and wireless communications. Many blind deconvolution
techniques, that are mostly tailored for particular applications,
have been proposed in these communities. These techniques can be divided
into two categories based on their general formulation of the problem.
The methods of the first category typically reduce the blind deconvolution
problem to a regularized least squares problem without imposing stochastic
models on either of the convolved signals. High computational cost
and sensitivity to noise are the main challenges for these methods.
The second category of blind deconvolution methods follow a Bayesian
approach and consider prior distributions for either or both of the
signals. An extensive review of the classic blind deconvolution methods
in imaging can be found in \cite{campisi_blind_2007}. A survey of
the multichannel blind deconvolution methods used in communications
can be found in \cite{tong_multichannel_1998} as well.

\subsection{Contributions}

In recent years, various imaging architectures have been proposed
that are based on randomly coded apertures. For example, in \cite{raskar_coded_2006},
rapid switching of exposure of cameras in random patterns is leveraged for motion deblurring. 
Furthermore, in the ``single-pixel-camera''
\cite{duarte_single-pixel_2008} and compressive hyperspectral imaging
systems \cite{sun_compressive_2009,kittle_multiframe_2010,li_compresive_2012,rajwade_coded_2013},
randomly coded masks are utilized to radically reduce the cost associated
with spatial or spectral sampling in traditional imaging systems.
In this paper, we consider the blind deconvolution problem in a similar
imaging architecture that relies on randomly coded masks. As illustrated
in Figure \ref{fig:MaskedImaging}, the considered imaging system
first generates different modulated copies of the target image, each
of which is then blurred by a fixed unknown filter (i.e., the blurring
kernel), and finally subsampled by a relatively small number of sensors.
 Throughout the paper, we only consider a
uniform subsampling operator in our model. This idealized
subsampling operator corresponds to an array of single pixel sensors
with uniform spacing. However, a broader class of subsampling schemes
can be treated in a similar fashion. For example, a more realistic
model for a single sensor is a weighted integrator of a few neighboring
pixels. Spatially uniform subsampling with an array of these types
of sensors can be reduced to the pointwise subsampling by absorbing
the weight window into the blur filter. We believe that our analysis
can be extended to even broader class of subsampling schemes, but
we do not pursue these extensions in this paper. 
 
As will be discussed in Section \ref{sec:MainResults}, in the absence
of any model for the blurring kernel the subsampling operation renders
the unique recovery of the image impossible, regardless of the number
of measurements acquired. Intuitively, there are two competing requirements for the blurring kernel in the considered imaging system. First, it is necessary that the image is blurred and spread to the extent that the subsampling sensors do not miss any pixel of the image. As will be seen, this requirement can be satisfied by considering a subspace model for the blurring kernel. Second, the blurring should not be excessive, otherwise the sampling becomes inefficient as different sensors collect (nearly) identical measurements. As will be seen in Section \ref{ssec:ConvexBD}, the severity of the blur can be quantified by the \emph{coherence} of the blurring kernel. Specifically, our results show that the coherence critically affects the sufficient number of masks for successful reconstruction through convex optimization.

We first study identifiability
of the problem without restricting the number of measurements. In
Section \ref{ssec:Identifiability}, using the results of \cite{kiraly_combinatorial_2012}
for matrix completion we express a necessary and sufficient condition
for identifiability which has a combinatorial nature.  Often, optical models can provide us with some crude approximations of the blurring kernel that can be leveraged as prior information. For example, we can conveniently incorporate these prior information in a subspace model where the linear span of the available approximations is considered as the set of feasible blurring kernels. Our second and more concrete identifiability result is obtained by considering a low-dimensional subspace model for the blurring kernel. We show that the described blind deconvolution problem is identifiable, if the mentioned low-dimensional subspace obeys certain conditions. In particular, our results show that if the blurring kernel has a sufficiently narrow ``bandwidth'' then the desired condition holds and thus we can uniquely identify the image and the blurring kernel. 

In the second part of our work, we show that, under a ``bandpass''
blur model, we can perform the blind deconvolution through lifting
and nuclear norm minimization. This systematic approach applies not only to our blind deconvolution problem, but also to a variety of other \emph{bilinear inverse problems} that involve unknown linear operators. The theoretical guarantees, which are explained in Section \ref{ssec:ConvexBD},
rely on construction of a \emph{dual certificate} for the nuclear norm minimization problem via the \emph{golfing scheme} \cite{gross-recovering-2011}.
Furthermore, the concentration inequalities recently developed in
the field of random matrix theory are frequently used throughout the
derivations. Finally, while we state our results under the bandpass
modelling of the filter, with some effort similar results can be established
for the more general subspaces described by the identifiability sufficient
conditions.

\subsection{Related work}
In \cite{candes_phase_2013}, the \emph{PhaseLift} method \cite{candes-phaselift-2013}
is extended to address the \emph{phase retrieval} problem in coded
diffraction imaging, where the measurements have a more intricate
structure. It is shown that the trace minimization (i.e., PhaseLift)
can solve the phase retrieval problem, if the randomly coded masks
follow certain ``admissible'' distributions and their number is
poly-logarithmic in the ambient dimension. The use of coded masks
in the phase retrieval problem of \cite{candes_phase_2013} is effectively
similar to that in the blind deconvolution problem we address in this
paper. However, the measurement model in this paper is different from
that of \cite{candes_phase_2013}.

In \cite{ahmed_blind_2014}, a convex programming technique
is proposed for blind deconvolution, where by \emph{lifting} the signal
and the filter to their outer product, the problem is cast as reconstruction
of a rank-one matrix from a set of linear measurements. It is shown
in \cite{ahmed_blind_2014} that the \emph{nuclear norm minimization}
can robustly and accurately recover the rank-one solution to the convolution
equations. This blind deconvolution technique imposes certain low-dimensional
subspace structures on the input and channel to reach a well-posed
problem. 

More recently, \cite{tang_convex_2014} has examined the problem of
blind deconvolution in an imaging system similar to what considered
in this paper. It is shown in \cite{tang_convex_2014} that one can
recover the image and the blurring kernel through lifting and nuclear norm
minimization, provided that the number of applied masks is greater
that the \emph{coherence} of the blurring kernel by a poly-logarithmic
factor of the image length. The fact that we consider the effect of
subsampling makes the imaging model considered
in this paper more general than that of \cite{tang_convex_2014}.
In the special case that subsampling is not applied, our problem reduces
to that of \cite{tang_convex_2014}. In this regime, the sufficient number of masks obtained here and in
\cite{tang_convex_2014} have similar growth order. Only for blurring
kernels that have more uniform spectrum (i.e., have low coherence),
our bounds can be slightly worse up to a double logarithmic factor
of the image length. We believe that our bounds can be further improved
using sharper inequalities throughout the derivations, but we do not
attempt to find the optimal logarithmic factors.

\subsection{Notation}

Throughout this paper we use the following notation. Matrices and
vectors are denoted by bold capital and small letters, respectively. A superscipt asterisk denotes the Hermitian transpose of matrices
and vectors (e.g., $\mb X^{*}$, $\mb x^{*}$). More generally, the
same symbol denotes the adjoint of linear operators (e.g., $\mc A^{*}$).
Restriction of a matrix $\mb X$ to the columns enumerated by an index
set $\msf J$ is denoted by $\mb X_{\msf J}$. For a vector $\mb x$,
we write $\left.\mb x\right|_{\msf J}$ to denote its restriction to the entries
indicated by the index set $\msf J$. Real and imaginary parts of
complex variables are denoted by preceding symbols $\Re$ and $\Im$,
respectively. Scalar functions applied to vectors or matrices act entrywise. Nullspace and range of linear operators are denoted
by $\mr{null}\left(\cdot\right)$ and $\mr{range}\left(\cdot\right)$,
respectively. Hadamard product (i.e., entrywise product) operation
on two matrices or vectors is denoted by $\odot$ symbol. Entrywise
conjugate of a matrix (or vector) is denoted by putting a bar above
the variable (e.g., $\overline{\mb X}$ is the entrywise conjugate
of $\mb X$). We frequently use the normalized Discrete Fourier Transform
(DFT) matrix which is denoted by $\mb F$ whose size should be clear
from the context. Furthermore, $\mb F_{n}$ is used to denote for
the restriction of $\mb F$ to its first $n$ rows. Moreover, the
$l$-th column of $\mb F^{*}$ is denoted by $\mb f_{l}$. The DFT
of a vector is denoted by the same name with a hat sign atop (e.g.,
for $\mb x\in\mbb C^{L}$, the DFT of $\mb x$ is denoted by $\widehat{\mb x}=\sqrt{L}\mb F\mb x$).The
diagonal matrix whose diagonal entries form a vector $\mb x$ is denoted
by $\mb D_{\mb x}$. Furthermore, the matrix of diagonal entries of
a square matrix $\mb X$ is denoted by $\mr{diag}\left(\mb X\right)$.
The vector norms $\left\Vert \cdot\right\Vert _{p}$ for $p\geq1$
are the standard $\ell_{p}$-norms. The spectral norm, the Frobenius
norm, and the nuclear norm are denoted by $\left\Vert \cdot\right\Vert $,
$\left\Vert \cdot\right\Vert _{F}$, and $\left\Vert \cdot\right\Vert _{*}$,
respectively. We find it convenient to use the expression $f\overset{\beta}{\gtrsim}g$
(or $f\overset{\beta}{\lesssim}g$) as a shorthand for inequalities
of the form $f\geq c_{\beta}g$ (or $c_{\beta}f\leq g$), where $c_{\beta}>0$
is some absolute constant that depends only a parameter $\beta$.
We drop the superscript in this notation whenever the constant factors
do not depend on any parameter.

\section{\label{sec:ProblemSetup}Problem setup}

\begin{figure} 
\noindent 
\centering
\includegraphics[height=0.33\textheight]{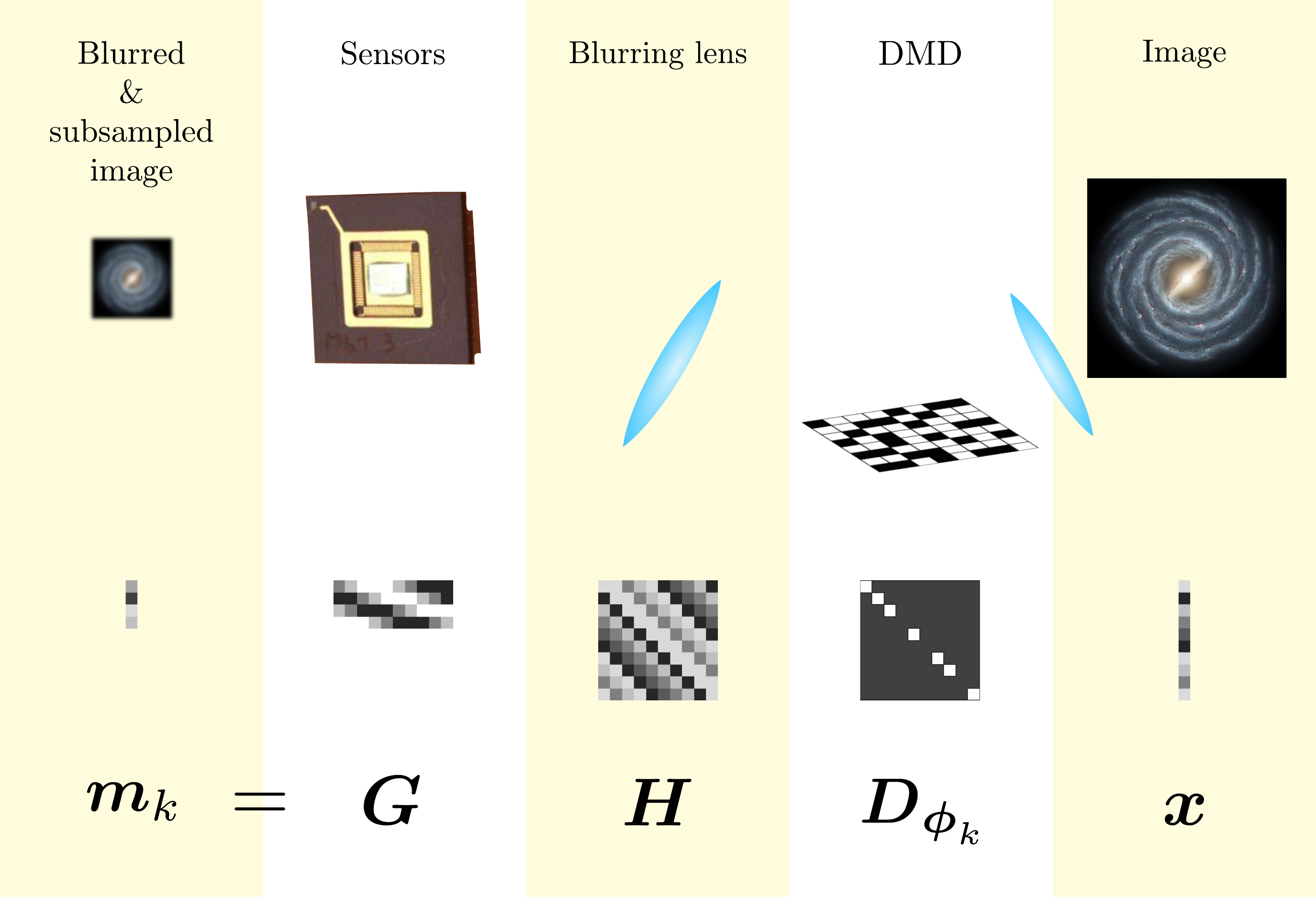}
\protect\caption{Schematic of the masked imaging system. The reflection of the target image from a DMD with a random pattern is blurred by a lens and then subsampled by a few sensors. The result is one set of measurements that correspond to the chosen DMD pattern.}
\label{fig:MaskedImaging}
\end{figure}

We consider a blind deconvolution problem in an imaging system depicted by Figure \ref{fig:MaskedImaging} that involves subsampling.
To simplify the exposition, through out the paper we consider 1D blind deconvolution, but generalization to 2D models is straightforward.
In our model, multiple binary masks $\mb{\phi}_{k}$ ($k=1,2,\dotsc,K$)
are applied to an image represented by $\mb x\in\mbb R^{L}$ one at
a time by the means of a Digital Micromirror Device (DMD). The masked reflection of the image
 from the DMD is then blurred through a secondary lens
represented by a filter $\mb h\in\mbb R^{L}$. We model the action
of the filter on its input by a circular convolution. The masked and
blurred image is then subsampled using $N\leq L$ sensors. Mathematically, the
described system can be represented by the equations 
\begin{align*}
\mb m_{k} & =\mb G\mb H\mb D_{\mb{\phi}_{k}}\mb x,
\end{align*}
 where $\mb H$ is a circulant matrix whose first column is $\mb h$
which models the blurring lens, $\mb G$ is an $N\times L$ matrix
representing the (linear) subsampling operator, and $\mb m_{k}$ is
the $k$-th measurement, corresponding to the $k$-th mask $\mb{\phi}_{k}$.
In this paper, we exclusively consider uniform pointwise subsampling as our
subsampling operator $\mb G$. Note that the above equation for all
of the measurements can be written compactly as 
\begin{align}
\mb M & =\mb G\mb H\mb D_{\mb x}\mb{\Phi},\label{eq:BlindDeconvolution}
\end{align}
where $\mb{\Phi}=\left[\arraycolsep=0.5ex\begin{array}{cccc}
\mb{\phi}_{1} & \mb{\phi}_{2} & \dotsm & \mb{\phi}_{K}\end{array}\right]$ is the matrix of masks and $\mb M=\left[\arraycolsep=0.5ex\begin{array}{cccc}
\mb m_{1} & \mb m_{2} & \dotsm & \mb m_{K}\end{array}\right]$ is the matrix of measurements. 

Accurate estimates of the blurring kernel $\mb h$ might not be available
in practice. For example, random vibrations of the lens and the subject relative to each other, turbulence of the propagation medium, or errors in measuring the focal length of the lens can preclude accurate estimation of the blur. Therefore, it is highly desirable to perform a blind deconvolution for reconstruction of both the image and the blurring
kernel from the measurements of the form (\ref{eq:BlindDeconvolution})
up to the global scaling ambiguity. 

Because $\mb G$ is in general a wide matrix, recovery of $\mb h$
and $\mb x$ (up to a scaling factor) can be ill-posed even with an
unlimited number of masks. Therefore, it is worthwhile to study the
identifiability of our inverse problem under the assumption $\mb{\Phi=\mb I}$.
In Section \ref{ssec:Identifiability}, we elaborate on the conditions
under which we can guarantee identifiability.

In Section \ref{ssec:ConvexBD} we introduce a convex program as a
systematic method for the blind deconvolution. To analyze this method,
we assume that the blurring kernel follows a ``bandpass'' model
that was suggested by the sufficient identifiability conditions. In
particular, in Section \ref{ssec:ConvexBD} we assume that $\mb h=\frac{1}{\sqrt{L}}\mb F_{N}\widehat{\mb h}$
is the blurring kernel for some $N$-dimensional vector $\widehat{\mb h}$.
Furthermore, to have a realistic model of the system, we assume that
the number of masks is limited and should be relatively small. While
ideal binary masks are $\left\{ 0,1\right\} $-valued, for technical
reasons we consider the elements of $\mb{\Phi}$ to be iid Rademacher
random variables that take values in $\left\{ \pm1\right\} $ with
equal probability. Note that this assumption is not unrealistic as
the $\left\{ 0,1\right\} $-valued masks can be converted to $\left\{ \pm1\right\} $-valued
masks by using an extra all-one mask.

\section{\label{sec:MainResults}Main results}

\subsection{\label{ssec:Identifiability}Identifiability without measurement
limitations}

In this section we analyze the identifiability of the image and blurring
kernel when arbitrarily large number of measurements are available.
Therefore, we can assume that $\mb{\Phi}$ is full-rank and has at
least as many columns as rows. This assumption implies that we can
reduce our observation model to 
\begin{align}
\mb M & =\mb G\mb H\mb D_{\mb x},\label{eq:Identifiability}
\end{align}
which is equivalent to \eqref{eq:BlindDeconvolution} for $\mb{\Phi=\mb I}$.
As mentioned in Section \ref{sec:ProblemSetup}, the subsampling matrix
$\mb G$ is assumed to model a uniform pointwise subsampling.
Therefore, each row of $\mb G$ is zero except at one entry where
it is one. This implies that each entry of the matrix of observations,
$\mb M$, can be expressed as $x_{i}h_{j}$ for certain indices $i$
and $j$. It is necessary to assume that the columns of $\mb G\mb H$
are all non-zero to ensure the information of every pixel of the image
is retained. Moreover, the observations can also be written as 
\begin{align}
\mr{vec}\left(\mb M\right) & =\left[\arraycolsep=0.5ex\begin{array}{cccc}
x_{1}\mb G_{1}^{\mr T} & x_{2}\mb G_{2}^{\mr T} & \dotsc & x_{L}\mb G_{L}^{\mr T}\end{array}\right]^{\mr T}\mb h,\label{eq:SwapImageBlur}
\end{align}
 where $\mr{vec}\left(\mb M\right)$ is the columnwise vectorization
of $\mb M$, $\mb G_{1}=\mb G$, and for $i>1$, $\mb G_{i}$ is obtained
by circularly shifting the columns of $\mb G_{i-1}$ to the left.
If any of the columns of the matrix on the right-hand side of (\ref{eq:SwapImageBlur})
is zero, the corresponding entry of $\mb h$ cannot be recovered from
the observations. Therefore, it is necessary to assume that the columns of 
the mentioned matrix are all non-zero. For the special choice of $\mb G$
that we consider, these two assumptions imply that the measurements
$x_{i}h_{j}$ for any particular $i$ and similarly for any particular
$j$ cannot be simultaneously zero.

\begin{figure} 
\noindent 
\centering
\includegraphics[height=0.25\textheight]{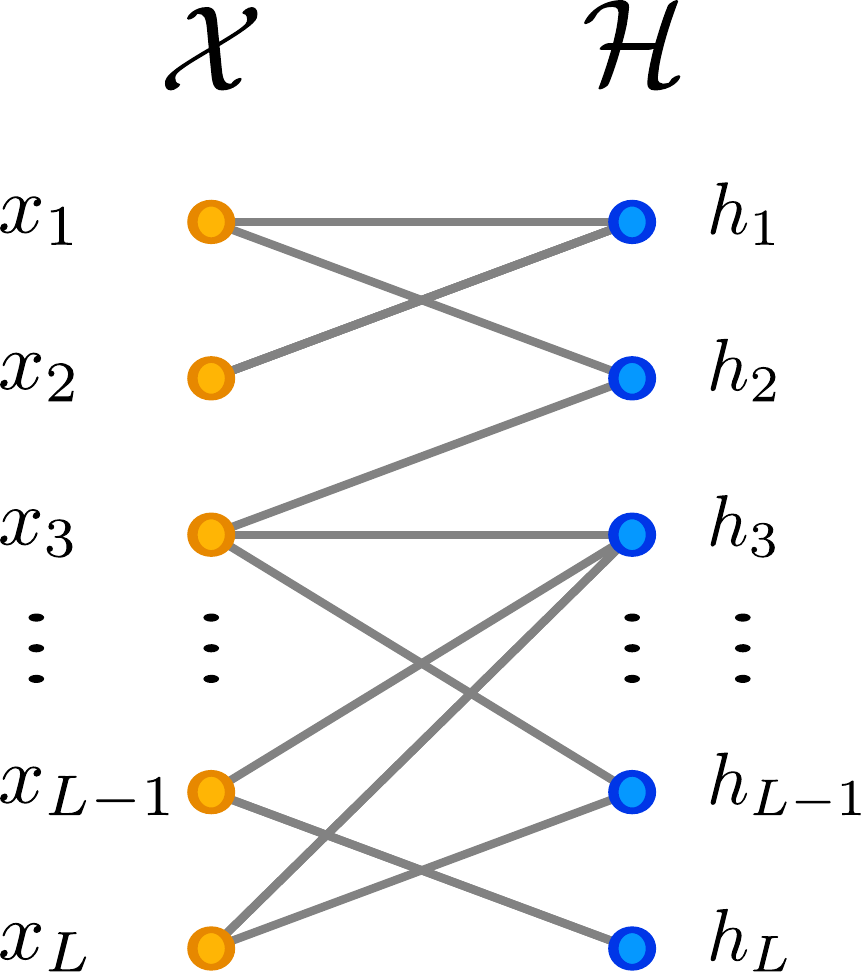}
\protect\caption{The bipartite graph constructed based on the observations. Only pairs of vertices are connected whose product is observed and non-zero.}
\label{fig:BipartiteGraph}
\end{figure}

The necessary and sufficient condition for identifiability of our
blind deconvolution problem is a simple special case of the combinatorial
identifiability conditions presented in \cite{kiraly_combinatorial_2012}
for the well-known \emph{low-rank matrix completion} problem. For
completeness, we state the identifiability condition in Lemma \ref{lem:IdentifiabilityGraph}
whose proof is subsumed in the appendix. Let $\mc G=\left(\mc X,\mc H,\mc E\right)$
be an undirected bipartite graph. The vertex partitions $\mc X=\left\lbrace x_{1},x_{2},\dotsc,x_{L}\right\rbrace $
and $\mc H=\left\lbrace h_{1},h_{2},\dotsc,h_{L}\right\rbrace $ correspond
to the entries of $\mb x$ and $\mb h$, respectively. Furthermore,
$\mc G$ is constructed such that $\left\lbrace x_{i},h_{j}\right\rbrace \in\mc E$
iff the value $x_{i}\cdot h_{j}$ is observed and is non-zero. An
example of such graphs is shown in Figure \ref{fig:BipartiteGraph}.
\begin{lem}
\label{lem:IdentifiabilityGraph}The rank-one matrix $\mb h\mb x^{\mr T}$
is uniquely recoverable from its subsampled entries iff the corresponding
bipartite graph has only one connected component of order greater
than one.
\end{lem}

Suppose that $\mb G$ models a uniform subsampling with period $T<L$
in \eqref{eq:Identifiability}. Then, as illustrated in Figure \ref{fig:MatrixCompletion},
the measurements in \eqref{eq:Identifiability} are identical to
the skew diagonal entries of the rank-one matrix $\mb h\mb x^{\mr T}$
that are $T$ entries apart in each row (or column). Therefore, our
deconvolution problem is basically a special rank-one matrix completion
problem where Lemma \ref{lem:IdentifiabilityGraph} applies. As an
illustrative example, consider the case that neither $\mb x$ nor
$\mb h$ have zero entries. The graph associated with the measurements
\eqref{eq:Identifiability} is then an $N$-regular bipartite graph
where $N=\left\lfloor \tfrac{L-1}{T}\right\rfloor +1$ is the number
of sampling sensors (i.e., the number the rows of $\mb G$). If we
also have $\gcd\left(T,L\right)=1$, then it is straightforward to
verify that the constructed graph is connected and by Lemma \ref{lem:IdentifiabilityGraph}
the matrix $\mb h\mb x^{\mr T}$ can be recovered uniquely.

\begin{figure}
\noindent 
\centering
\def\svgwidth{0.5\columnwidth}
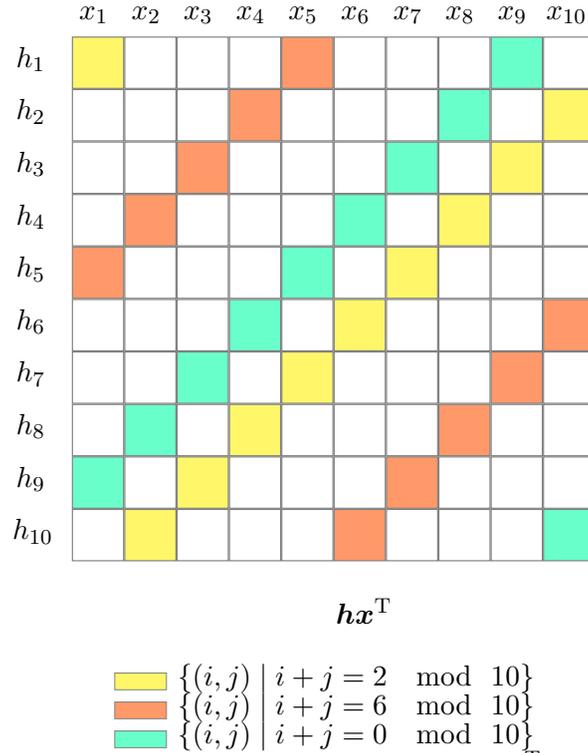
\protect\caption{Measurements given by \eqref{eq:Identifiability} as the entries of $\mb{h}\mb{x}^\mr{T}$ for $L=10$ and $T=4$.}
\label{fig:MatrixCompletion}
\end{figure}

Although Lemma \ref{lem:IdentifiabilityGraph} establishes the necessary
and sufficient condition for identifiability of our problem, it is
desirable to have alternative guarantees that are not combinatorial
in nature. The following theorem provides a sufficient condition for
identifiability by imposing a subspace structure for the blurring
kernel.
\begin{thm}
\label{thm:IdentifiabilitySubspace}For $N=\left\lfloor \tfrac{L-1}{T}\right\rfloor +1$
let $\mb V\in\mbb C^{L\times N}$ be a given matrix whose restriction
to rows indexed by 
\begin{align*}
\msf J_{i} & :=\left\lbrace j\mid 1\leq j\leq L\text{ and }j=-i+2+kT\mod L\text{ for some }0\leq k\leq N-1\right\rbrace ,
\end{align*}
 is full-rank for all $i=1,2,\dotsc,L$. For any image $\mb x\neq\mb 0$
and any blurring kernel $\mb h\in\mr{range}\left(\mb V\right)$, the
rank-one matrix $\mb h\mb x^{\mr T}$ can be uniquely recovered as
the solution to the blind deconvolution problem \eqref{eq:Identifiability}. \end{thm}

As can be inspected in Figure \ref{fig:MatrixCompletion}, for each
$i=1,2,\dotsc,L$ the set $\msf J_{i}$ described in Theorem \ref{thm:IdentifiabilitySubspace}
determines the location of observed entries in the $i$-th column
of the matrix $\mb h\mb x^{\mr T}$. Using this observation, the proof
of Theorem \ref{thm:IdentifiabilitySubspace}, provided in
Appendix \ref{ssec:Identifiability}, is straightforward.
\begin{cor}
\label{cor:LowPassBD}Let $\mb F_{\Omega}$ denote a matrix of some
$N$ (circularly) consecutive rows of the normalized $L$-point DFT
matrix that are indexed by $\Omega$. For any image $\mb x\neq\mb 0$
and blurring kernel $\mb h\in\mr{range}\left(\mb F_{\Omega}^{*}\right)$,
we can recover the rank-one matrix $\mb h\mb x^{\mr T}$ uniquely
from the measurements given by \eqref{eq:Identifiability}. \end{cor}
\begin{proof}
Without loss of generality we assume that $\Omega=\left\{ 1,2,\dotsc,N\right\} $
as the proof is similar for other valid choices of $\Omega$. The
result follows immediately from Theorem \ref{thm:IdentifiabilitySubspace}
should the matrix $\mb V=\mb F_{\Omega}^{*}=\mb F_{N}^{*}$ satisfies
the requirements of the theorem. Namely, it suffices to show that
the restriction of $\mb F_{N}^{*}$ to the rows indexed by $\msf J_{i}$
is full-rank for all $i=1,2,\dotsc,L$. With $\omega:=e^{2\pi\mr i/L}$
the restriction of $\mb F_{N}^{*}$ to the rows in $\msf J_{i}$ can
be written as 
\begin{align*}
\mb F_{\msf J_{i},N}^{*} & :=\frac{1}{\sqrt{L}}\left[\arraycolsep=0.5ex\begin{array}{cccc}
1 & \omega^{-i+1} & \dotsm & \omega^{\left(N-1\right)\left(-i+1\right)}\\
1 & \omega^{T-i+1} & \dotsm & \omega^{\left(N-1\right)\left(T-i+1\right)}\\
\vdots & \vdots & \ddots & \vdots\\
1 & \omega^{\left(N-1\right)T-i+1} & \dotsm & \omega^{\left(N-1\right)\left(\left(N-1\right)T-i+1\right)}
\end{array}\right].
\end{align*}
 Having $\mb F_{\msf J_{i},N}\mb a=\mb 0$ for some $\mb a\in\mbb C^{N}$
is equivalent to having the polynomial $a\left(z\right):=\sum_{i=1}^{N}a_{i}z^{i-1}$
vanishing at $z=\omega^{-i+1},\omega^{T-i+1},\dotsc,\omega^{\left(N-1\right)T-i+1}$
which are $N$ distinct points in $\mbb C$. This is possible only
if $a\left(z\right)\equiv0$ because the degree of $a\left(z\right)$
is less than $N$. Therefore, $\mb F_{\msf{J}_{i,N}}^{*}\mb a=\mb 0$
iff $\mb a=\mb 0$ as desired.
\end{proof}

While Corollary \ref{cor:LowPassBD} shows that unique reconstruction
of the image and the blurring kernel is possible for ``bandpass''
blurring kernels, it does not provide any robust recovery method.
Interestingly, there is also a robust recovery method for the bandpass
model as described below. As in the proof of the corollary, we consider
the case of $\Omega=\left\{ 1,2,\dotsc,N\right\} $ to simplify the
exposition. Note that the measurement can be written as 
\begin{align*}
\mb M & =\mb G\mb H\mb D_{\mb x}+\mb E=\mb G\mb F_{N}^{*}\mb D_{\widehat{\mb h}}\mb F_{N}\mb D_{\mb x}+\mb E,
\end{align*}
 where 
 \begin{quote}
 \textit{$\widehat{\mb h}$, with slight abuse of our notation, denotes
the frequency content of the filter $\mb h$ (i.e., $\mb h=\frac{1}{\sqrt{L}}\mb F_{N}^{*}\widehat{\mb h}$),}
\end{quote}
and $\mb E$ denotes the measurement error. Since $\mb G$ is assumed
to be a uniform subsampling operator, the matrix $\widetilde{\mb G}:=\mb G\mb F_{N}^{*}$
is invertible as shown in the proof of Corollary \ref{cor:LowPassBD}.
Therefore, we can write 
\begin{align*}
\widetilde{\mb G}^{-1}\mb M & =\mb D_{\widehat{\mb h}}\mb F_{N}\mb D_{\mb x}+\widetilde{\mb G}^{-1}\mb E=\mb F_{N}\odot\left(\widehat{\mb h}\overline{\mb x}^{*}\right)+\widetilde{\mb G}^{-1}\mb E.
\end{align*}
 Let $\overline{\mb F}_{N}$ be the entrywise conjugate of $\mb F_{N}$.
Entrywise multiplication of both sides
of the above equation by $\overline{\mb F}_{N}$ yields 
\begin{align*}
\overline{\mb F}_{N}\odot\left(\widetilde{\mb G}^{-1}\mb M\right) & =\frac{1}{L}\widehat{\mb h}\overline{\mb x}^{*}+\overline{\mb F}_{N}\odot\left(\widetilde{\mb G}^{-1}\mb E\right).
\end{align*}
 Therefore, we can estimate $\widehat{\mb h}\overline{\mb x}^{*}$
as the best rank-one approximation to the matrix \linebreak$L\overline{\mb F}_{N}\odot\left(\widetilde{\mb G}^{-1}\mb Y\right)$~with estimation error being less than $2L\left\Vert \overline{\mb F}_{N}\odot\left(\widetilde{\mb G}^{-1}\mb E\right)\right\Vert _{F}$.

\subsection{\label{ssec:ConvexBD}Blind deconvolution via nuclear norm minimization}

In this section we consider a convex programming approach for solving
(\ref{eq:BlindDeconvolution}) under the bandpass model for the
blurring kernel described in \ref{cor:LowPassBD}. Again for simplicity,
we only consider the case that $\mb h=\frac{1}{\sqrt{L}}\mb F^{*}_{N}\widehat{\mb h}$
with $\widehat{\mb h}$ being the (truncated) DFT of $\mb h$. 

Similar to the discussion following (\ref{cor:LowPassBD}) we can
rewrite the measurement equation (\ref{eq:BlindDeconvolution}) as

\begin{align*}
\mb M & =\mb G\mb H\mb D_{\mb x}\mb{\Phi}=\mb G\mb F_{N}^{*}\mb D_{\widehat{\mb h}}\mb F_{N}\mb D_{\mb x}\mb{\Phi}.
\end{align*}
 Since $\widetilde{\mb G}=\mb G\mb F_{N}^{*}$ is invertible (see
proof of Corollary \ref{cor:LowPassBD} above), it suffices to analyze
recoverability of $\widehat{\mb h}$ and $\mb x$ from observations
\begin{align*}
\widetilde{\mb M} & :=\sqrt{\frac{L}{K}}\widetilde{\mb G}^{-1}\mb M=\sqrt{\frac{L}{K}}\mb D_{\widehat{\mb h}}\mb F_{N}\mb D_{\mb x}\mb{\Phi}\\
 & =\sqrt{\frac{L}{K}}\left(\mb F_{N}\odot\left(\widehat{\mb h}\overline{\mb x}^{*}\right)\right)\mb{\Phi}
\end{align*}
Define the linear operator $\mc A:\mbb C^{N\times L}\to\mbb C^{N\times K}$
as 
\begin{align}
\mc A\left(\mb X\right) & :=\sqrt{\frac{L}{K}}\left(\mb F_{N}\odot\mb X\right)\mb{\Phi},\label{eq:MesurementOperator}
\end{align}
whose adjoint is given by 
\begin{align*}
\mc A^{*}\left(\mb Y\right) & =\sqrt{\frac{L}{K}}\overline{\mb F}_{N}\odot\left(\mb Y\mb{\Phi}^{*}\right).
\end{align*}
We have $\widetilde{\mb M}=\mc A\left(\widehat{\mb h}\overline{\mb x}^{*}\right)$.
Without loss of generality we assume that \begin{quote}\textit{
$\mb x$, the target image, and $\widehat{\mb h}$, the DFT of the
blurring kernel, both have unit $\ell_{2}$-norm.} 
\end{quote} Furthermore, we define the coherence of the blurring
kernel as 
\begin{align}
\mu & :=\frac{\left\Vert \widehat{\mb h}\right\Vert _{\infty}^{2}}{\left\Vert \mb h\right\Vert _{2}^{2}}=L\left\Vert \widehat{\mb h}\right\Vert _{\infty}^{2}.\label{eq:Coherence}
\end{align}
 We show that the nuclear norm minimization 

\begin{align}
\arg\min_{\mb X} & \left\Vert \mb X\right\Vert _{*}\label{eq:NucMin}\\
\mr{subject\ to} & \ \mc A\left(\mb X\right)=\widetilde{\mb M},\nonumber 
\end{align}
 can recover the matrix $\widehat{\mb h}\overline{\mb x}^{*}$ with
high probability. 
\begin{thm}
\label{thm:BlindDeconvolution}Let $\mb{\Phi}\in\left\{ \pm1\right\} ^{L\times K}$
be a random matrix with iid Rademacher entries and define the linear
operator $\mc A$ as in (\ref{eq:MesurementOperator}). Then, for
$K\overset{\beta}{\gtrsim}\mu\log^{2}L\,\log\frac{Le}{\mu}\,\log\log\left(N+1\right)$
we can guarantee that (\ref{eq:NucMin}) recovers $\widehat{\mb h}\overline{\mb x}^*$
uniquely, with probability exceeding $1-O\left(NL^{-\beta}\right)$.\end{thm}
\begin{rem}
Because $\mb h$ is assumed to have only $N$ active frequency components,
the coherence is bounded from below as 
\begin{align*}
\mu & \geq\frac{L}{N}.
\end{align*}
Therefore, the bound that the theorem imposes on the number of masks
can be simplified to 
\begin{align*}
K & \overset{\beta}{\gtrsim}\mu\log^{2}L\,\log\left(N+1\right)\,\log\log\left(N+1\right).
\end{align*}
 Furthermore, the result of Theorem \ref{thm:BlindDeconvolution}
suggests that $K\gtrsim\frac{L}{N}\log^{2}L\,\log\left(N+1\right)\,\log\log\left(N+1\right)$
random masks are necessary for \ref{eq:NucMin} to successfully recover
the target rank-one matrix. The dependence of this lower bound on
$L$ may seem unsatisfactory. However, if $K<\frac{L}{N}$, for any
fixed $\mb h$ the equation (\ref{eq:BlindDeconvolution}) will be
underdetermined with respect to $\mb x$, thereby $\mb x$ cannot
be recovered uniquely. Therefore, the number of measurements required
by Theorem \ref{thm:BlindDeconvolution} is suboptimal only by some
poly-logarithmic factors of $L$ and $N$.%

\end{rem}

\begin{rem}
To bring robustness to the proposed blind deconvolution approach,
we can modify (\ref{eq:NucMin}) by replacing the linear constraint
with an inequality of the form $\left\Vert \mc A\left(\mb X\right)-\widetilde{\mb M}\right\Vert _{F}\leq\delta$,
where $\widetilde{\mb M}$ denotes the noisy observations and $\delta$
is a constant that depends on the noise energy. Although accuracy
of the described convex program can be analyzed as well, we do not
attempt to derive these accuracy guarantees here and refer the interested
readers to \cite{candes-matrix-2010}, \cite{gross-quantum-2010},
and \cite{ahmed_blind_2014} for similar derivations.
\end{rem}

\section{\label{sec:Experiments}Numerical experiments}

For numerical evaluation of the blind deconvolution via (\ref{eq:NucMin})
we conducted two simulations using synthetic data. To solve
the nuclear norm minimization we used the solver proposed in \cite{burer_nonlinear_2003,burer_local_2005}.
In the first experiment we used an astronomical image of dimension $L=128\times128$ as the test image.\footnote{The image is adapted from NASA's Hubble Ultra Deep Field image that can be found online at:\linebreak \url{http://commons.wikimedia.org/wiki/File:Hubble\_ultra\_deep\_field\_high\_rez\_edit1.jpg}}
To generate the blur kernel we generated an $128\times128$ matrix
of iid standard normal random variables and then suppressed its 2D
DFT content outside a square of size $N=43\times43$ centered at the
origin.%
\footnote{The DFT indices are treated as integers modulo $128$.%
} The subsampling of the blurred image is performed at the rate of
$\frac{1}{3}$ both vertically and horizontally which provides $N$
scalar measurements per applied mask. We computed the subsampled convolution
for $K=300$ random Rademacher masks which yields a total of $K\times N=554700$
scalar measurements. The relative error between the target rank-one
matrix and the estimate obtained by (\ref{eq:NucMin}) is in the order
of $10^{-7}$. Figure \ref{fig:Hubble} also illustrates that the
proposed blind deconvolution method has successfully recovered the
normalized image and the normalized blurring kernel up to the prescribed
tolerance.

\begin{figure}
\begin{subfigure}[t]{1\textwidth}
\centering
\includegraphics[width=1\textwidth]{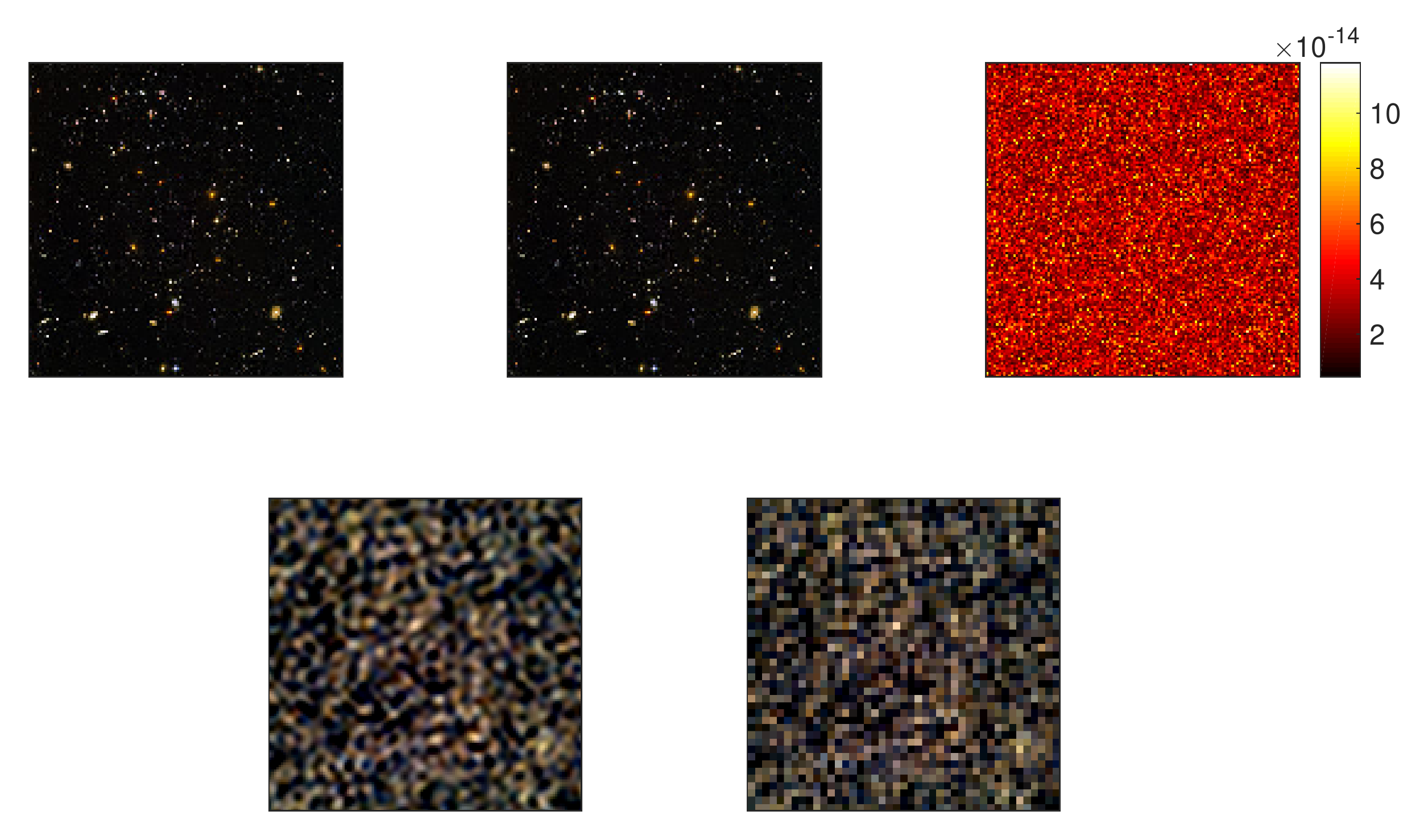}
\protect\caption{Top row: The original $128\times 128$ HUDF image (left), reconstructed image (center), and the difference between normalized images (right),\protect \\
Bottom row: blurred image (left), $3\msf{X}$ magnified subsampled blurred image (right)}
\label{fig:Hubble-a}
\end{subfigure}

\begin{subfigure}[t]{1\textwidth}  
\centering
\includegraphics[width=1\textwidth]{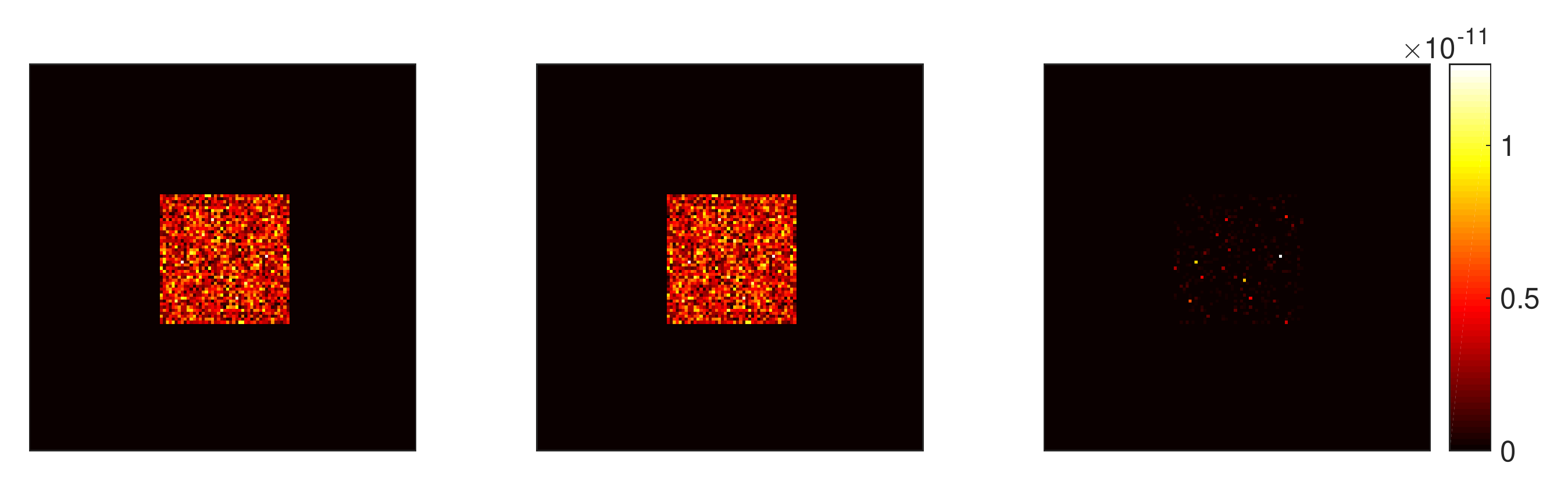} 
\protect\caption{Spectra of the original $128\times 128$ blur kernel (left), the reconstructed blur kernel (center), and the difference between normalized blur kernels (right)}  
\label{fig:Hubble-b}    
\end{subfigure}

\protect\caption{Blind deconvolution of a Hubble Ultra-Deep Field image and a synthetic
blur kernel}
\label{fig:Hubble}
\end{figure}

In the second experiment we considered a more realistic model for
the blur kernel.  We use eight $64\times64$ consecutive slices of
a 3D Point Spread Function (PSF) in the Born \& Wolf optical model
that is generated by the PSFGenerator package \cite{kirshner_psf_2013}
to create a subspace model for the target PSF. We used the default
model parameters set by the package, except for the size of the PSF
array in the XY plane mentioned above. Figure \ref{fig:subspace} depicts an orthonormal basis of the
subspace (in magnitude) that we used in the experiment. We chose one
of the original PSFs as our target PSF. Furthermore,
we use a $128\times128$ fluorescent microscopy image of \emph{endothelial
cells} as the target shown in Figure \ref{fig:image} (top left). For
this experiment, the number of applied Rademacher masks is $K=200$.
The subsampling is uniform in vertical and horizontal directions
at the rate of $\frac{1}{8}$. Therefore, the number of observations
per mask is $N=16\times16$. To have a reference for comparison, the
target image blurred by the target PSF and the $16\times16$ subsampled
version of the blurred image are shown in the bottom row of Figure \ref{fig:image}.

Figure \ref{fig:blur} illustrates the target PSF (left), the estimated
PSF (center), and the error between the normalized target and the
normalized estimated PSFs (right). Similarly, the top row of \ref{fig:image}
illustrates the target image (left), the estimated image (center),
and the error between the normalized target and the normalized estimated
PSFs (right). As can be seen in these figures, the proposed blind
deconvolution method has found accurate reconstructions of the
PSF and the image. The relative error in the lifted domain is also
in the order of  $10^{-6}$.

\begin{figure}
\noindent
\centering

\includegraphics[width=1\textwidth]{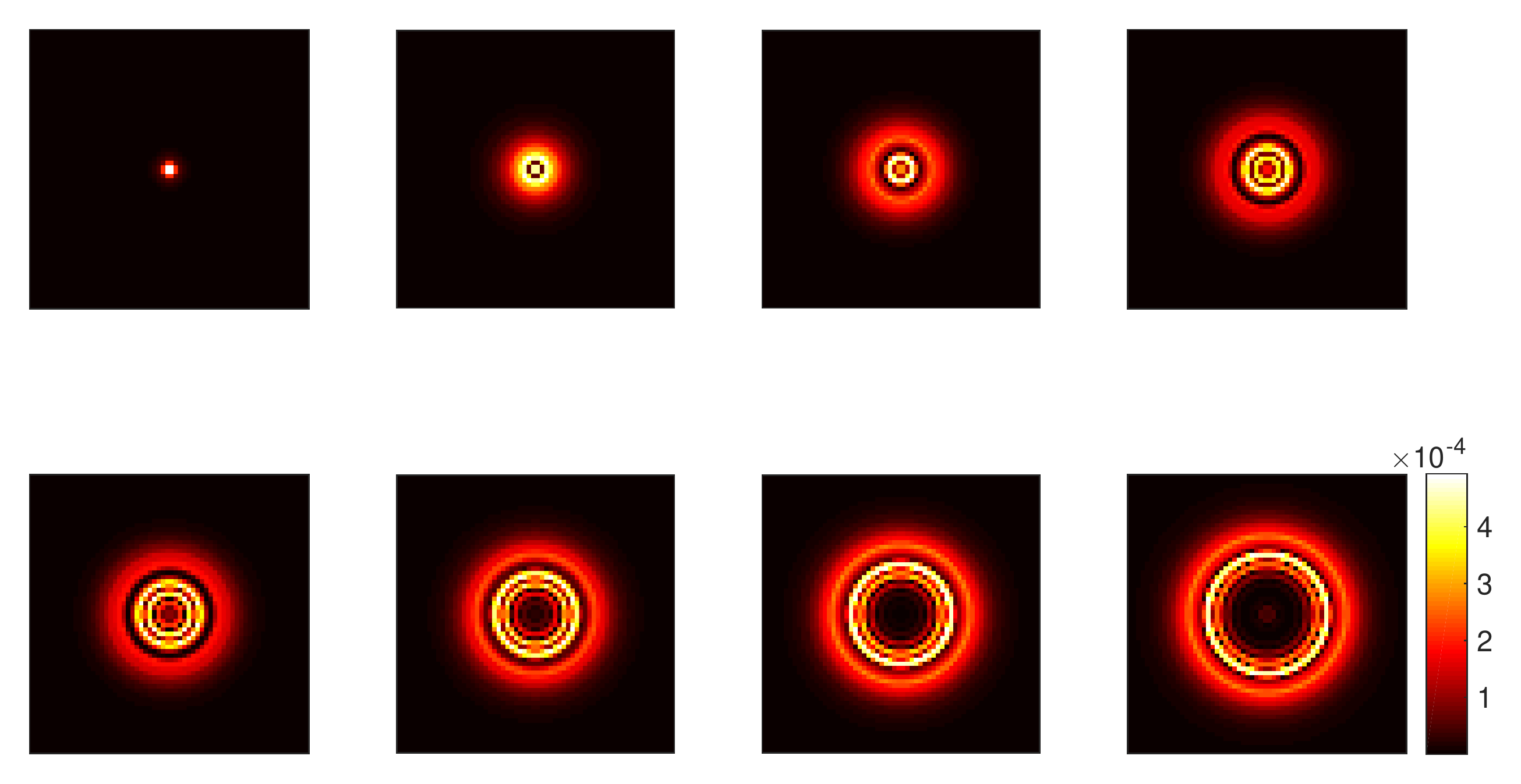}

\protect\caption{The orthonormal basis for the PSF subspace shown in magnitude}

\label{fig:subspace}
\end{figure}
\begin{figure}
\noindent
\centering

\includegraphics[width=1\textwidth]{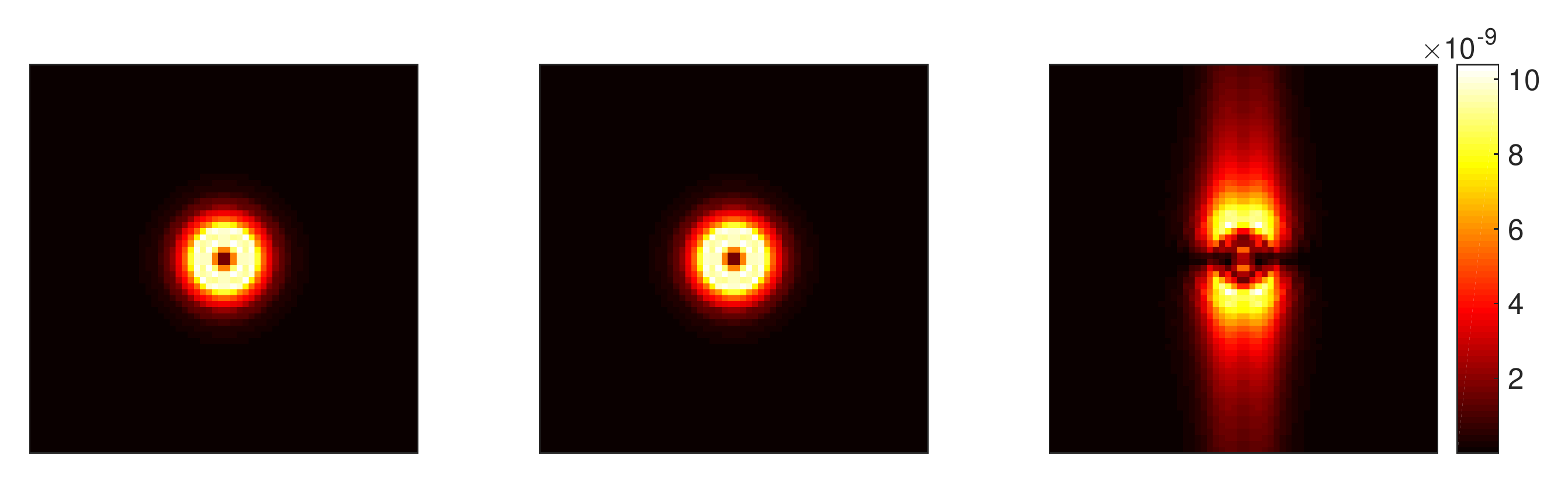}

\protect\caption{The target PSF (left), the reconstructed PSF (center), and the difference
between the normalized PSFs (right)}
\label{fig:blur}
\end{figure}
\begin{figure}
\noindent
\centering\includegraphics[width=1\textwidth]{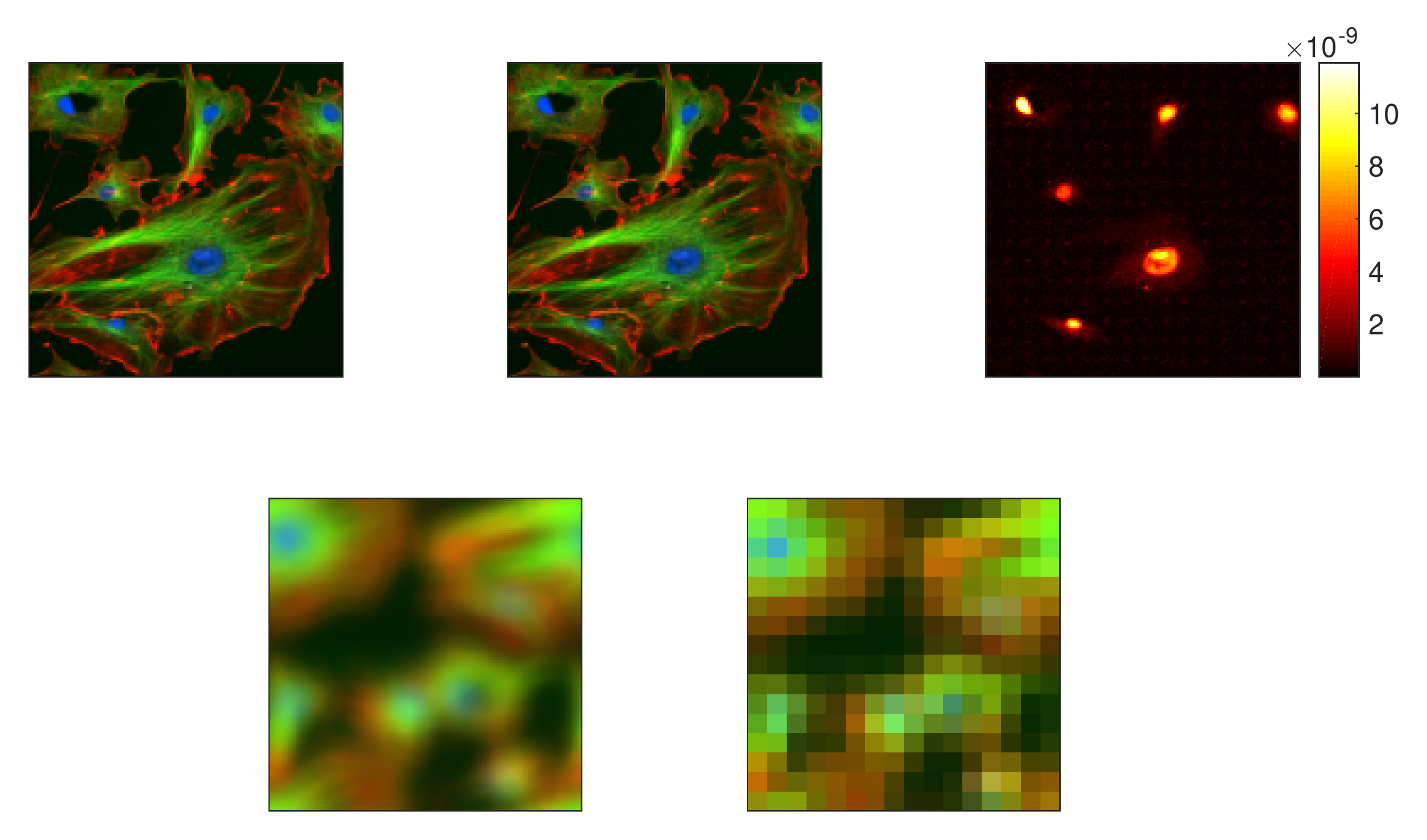}

\protect\caption{Top row: image of fluorescent \emph{endothelial} cells (left), reconstructed
image (center), and the difference between the normalized images (right),\protect \\
Bottom row: blurred image (left), $8\protect\msf X$ magnified subsampled
blurred image (right)}
\label{fig:image}
\end{figure}

\section{Conclusions}

In this paper we studied the blind deconvolution problem in an imaging system
that captures multiple instances of the target image modulated by
randomly coded masks, blurred, and then spatially subsampled. We first
expressed the necessary and sufficient condition for identifiability
of the problem using a graph representation of the unknowns and the measurements.
Furthermore, we formulate a different sufficient identifiability condition
by considering a subspace model for the blurring kernel. Finally,
under a special subspace model, namely the bandpass model, we derived
the sufficient number of random masks that allow successful blind
deconvolution through lifting and nuclear norm minimization.

An interesting extension to our work is to consider a subspace model
for the image rather than the blurring kernel. We believe that this
model would relax the identifiability requirements imposed on the
blurring kernel that might not be appropriate in certain scenarios.
The image subspace model can be further relaxed to sparsity
with respect to some basis (e.g., wavelet, Fourier, etc). This subspace-sparse
model would lead to estimation of simultaneously low-rank and row-sparse
matrices in the lifted domain. These problems, in their general form,
are known to be challenging and resistant to usual convex relaxation
techniques \cite{oymak_simultaneously_2015}. However, we believe
that the particular measurement model in the blind deconvolution problem can enable us
to perform accurate and efficient blind deconvolution.

\appendix

\section{Proofs of the results in Section \ref{ssec:Identifiability} }

In this section we provide the proofs pertaining to the identifiability
analysis provided in Section \ref{ssec:Identifiability}.
\begin{proof}[Proof of Lemma \ref{lem:IdentifiabilityGraph}]
 The assumptions made in Section \ref{ssec:Identifiability} ensure
that the isolated vertices in the considered bipartite graph correspond
to the entries with value zero. Furthermore, for each edge of the
graph we observe the product of its end nodes. Therefore, in each
of the connected components of the graph, choosing the value
of only one of the vertices is enough to uniquely determine the value of
the other vertices of that component. This assignment of values does
not depend on the other connected components of the graph. Therefore,
if there are two or more connected components of order greater than
one, each of them can take independent values. This implies that the
corresponding matrix $\mb h\mb x^{\mr T}$ cannot be recovered uniquely.
\end{proof}

\begin{proof}[Proof of Theorem \ref{thm:IdentifiabilitySubspace}]
As mentioned in Section \ref{ssec:Identifiability}, we have restricted
our problem to scenarios where the measurements $x_{i}h_{j}$ for
any particular $i$ and similarly for any particular $j$ cannot be
simultaneously zero. Therefore, the zero entries of $\mb x$ (and
also $\mb h$) can be easily identified from the zero measurements.
By the assumption that $\mb x\neq0$, there is at least one $1\leq i\leq L$
where $x_{i}\neq0$ . Let $\left.\mb h\right|_{\msf J_{i}}$ denote the restriction
of $\mb h$ to the entries indexed by $\msf J_{i}$. Therefore, we
observe the vector $x_{i}\left.\mb h\right|_{\msf J_{i}}$ which is nonzero.
Invoking the assumption that the restriction of $\mb V$ to the rows
indexed by $\msf J_{i}$ is full-rank we deduce that we can recover
the vector $x_{i}\mb h$, uniquely by solving a least squares problem.
Then, it is straightforward to recover any remaining nonzero $x_{l}\mb h$
from the measurements.
\end{proof}

\section{Proofs of the results in Section \ref{ssec:ConvexBD} }

\subsection{Tools from probability theory}

In this section we provide the definitions and results from probability
theory that we frequently apply in the proofs.
\begin{defn}
For a convex and non-decreasing function $\psi:\mbb{R}_+\to\mbb{R}_+$ that satisfies\linebreak $\psi\left(0\right)=0$, the
Orlicz $\psi$-norm for a random matrix (or vector) $\mb X$ is defined
as 
\begin{align*}
\left\Vert \mb X\right\Vert _{\psi} & =\inf\left\{ u>0\mid\mbb E\left[\psi\left(\frac{\left\Vert \mb X\right\Vert }{u}\right)\right]\leq1\right\} .
\end{align*}
\end{defn}
Some important special cases of the Orlicz $\psi$-norms are
\begin{itemize}
\item the Orlicz $2$-norm, also known as the sub-Gaussian norm, denoted
by $\left\Vert \cdot\right\Vert _{\psi_{2}}$ with $\psi_{2}\left(t\right)=e^{\frac{t^{2}}{2}}-1$,
and 
\item the Orlicz $1$-norm, also known as subexponential norm, denoted by
$\left\Vert \cdot\right\Vert _{\psi_{1}}$ with $\psi_{1}\left(t\right)=e^{t}-1$.\end{itemize}
\begin{prop}[{Matrix Bernstein's inequality \cite[Proposition 2]{koltchinskii_2011_nuclear}}]
\label{pro:MatrixBernstein's} Let $\mb X_{1},\mb X_{2},\dotsc,$
and $\mb X_{n}$ be independent random matrices of dimension $d_{1}\times d_{2}$
that satisfy $\mbb E\left[\mb X_{i}\right]=\mb 0$. Suppose that for
$B>0$ we have 
\begin{align*}
\max_{i=1,2,\dotsc,n}\left\Vert \mb X_{i}\right\Vert _{\psi_{1}} & \leq B,
\end{align*}
 and define 
\begin{align*}
\sigma^{2}: & =\max\left\{ \left\Vert \mbb E\left[\sum_{i=1}^{n}\mb X_{i}\mb X_{i}^{*}\right]\right\Vert ,\left\Vert \mbb E\left[\sum_{i=1}^{n}\mb X_{i}^{*}\mb X_{i}\right]\right\Vert \right\} .
\end{align*}
 Then, there exist a constant $C$ such that for all $t\geq0$ the
tail bound 
\begin{align*}
\left\Vert \sum_{i=1}^{n}\mb X_{i}\right\Vert  & \leq C\max\left\{ \sigma\sqrt{t+\log\left(d_{1}+d_{2}\right)},B\log\left(\frac{\sqrt{n}B}{\sigma}\right)\left(t+\log\left(d_{1}+d_{2}\right)\right)\right\} ,
\end{align*}
 holds with probability at least $1-e^{-t}$. 
\end{prop}

\begin{prop}[{Orlicz norm of a finite maximum \cite[Lemma 2.2.2]{vandervaart_weak_1996}}]
\label{pro:MaxOrlicz} Let $\psi$ be a convex, nondecreasing, nonzero
function that obeys $\psi\left(0\right)=0$ and $\lim\sup_{x,y\to\infty}\psi\left(x\right)\psi\left(y\right)/\psi\left(cxy\right)<\infty$
for some constant $c$. Then, for any random variables $x_{1},x_{2},\dotsc,$
and $x_{n}$ we have 
\begin{align*}
\left\Vert \max_{i=1,2,\dotsc,n}x_{i}\right\Vert _{\psi} & \leq c_{\psi}\psi^{-1}\left(n\right)\max_{i=1,2,\dotsc,n}\left\Vert x_{i}\right\Vert _{\psi},
\end{align*}
 for some constant $c_{\psi}$ depending only on $\psi$.
\end{prop}

\begin{prop}[{Hanson-Wright inequality \cite[Theorem 1.1]{rudelson_hanson-wright_2013}}]
\label{pro:Hanson-Wright}Let $\mb x\in\mbb R^{n}$ be a random vector
with independent components $x_{i}$ which satisfy $\mbb E\left[x_{i}\right]=0$
and $\left\Vert x_{i}\right\Vert _{\psi_{2}}\leq\rho$. Let $\mb A$
be an $n\times n$ matrix. Then, for every $t\geq0$, 
\begin{align*}
\mbb P\left\{ \left|\mb x^{\mr T}\mb A\mb x-\mbb E\left[\mb x^{\mr T}\mb A\mb x\right]\right|>t\right\}  & \leq2e^{-c\min\left\{ \frac{t^{2}}{\rho^{4}\left\Vert \mb A\right\Vert _{F}^{2}},\frac{t}{\rho^{2}\left\Vert \mb A\right\Vert }\right\} }.
\end{align*}

\end{prop}

\subsection{Proof of Theorem \ref{thm:BlindDeconvolution}}

The auxiliary and intermediate lemmas needed for the proof of Theorem
\ref{thm:BlindDeconvolution} are subsumed to the latter parts of
this section. To prove the theorem we need to construct a dual certificate that
exhibits certain properties on the ``support set'' 
\begin{align}
\msf T & =\left\{ \widehat{\mb h}\mb v^{*}+\mb u\overline{\mb x}^{*}\mid\mb u\in\mbb C^{N},\mb v\in\mbb C^{L}\right\} ,\label{eq:Support}
\end{align}
 and its orthogonal complement $\msf T{}^{\perp}$.
\begin{proof}[Proof of Theorem \ref{thm:BlindDeconvolution}]
 Our proof begins by stating the conditions for unique recovery of
$\widehat{\mb h}\overline{\mb x}^{*}$ from (\ref{eq:NucMin}) which
parallels the arguments in \cite[Section 3.1]{ahmed_blind_2014}.
Without repeating every detail explained in \cite{ahmed_blind_2014},
we provide a sketch here for clarification. Recall that $\widehat{\mb h}$
and $\mb x$ are assumed to have unit $\ell_{2}$-norm. Let 
\begin{align*}
\mc P_{\msf T}: & \mb S\mapsto\widehat{\mb h}\widehat{\mb h}^{*}\mb S+\mb S\overline{\mb x}\overline{\mb x}^{*}-\widehat{\mb h}\widehat{\mb h}^{*}\mb S\overline{\mb x}\overline{\mb x}^{*},\qquad & \mc P_{\msf T^{\perp}}: & \mb S\mapsto\left(\mb I-\widehat{\mb h}\widehat{\mb h}^{*}\right)\mb S\left(\mb I-\overline{\mb x}\overline{\mb x}^{*}\right)
\end{align*}
 denote the orthognoal projections onto the subspaces $\msf T$ and
$\msf T^{\perp}$, respectively. It can be shown that (see,
e.g. \cite{recht_2011_simpler}) the matrix $\widehat{\mb h}\overline{\mb x}^{*}$
is a unique minimizer of (\ref{eq:NucMin}) if there exists a matrix
$\mb Y\in\mr{range}\left(\mc A^{*}\right)$ such that 
\begin{align*}
\Re\left\langle \widehat{\mb h}\overline{\mb x}^{*}-\mc P_{\msf T}\left(\mb Y\right),\mc P_{\msf T}\left(\mb Z\right)\right\rangle -\Re\left\langle \mc P_{\msf T^{\perp}}\left(\mb Y\right),\mc P_{\msf T^{\perp}}\left(\mb Z\right)\right\rangle +\left\Vert \mc P_{\msf T^{\perp}}\left(\mb Z\right)\right\Vert _{*} & >0,
\end{align*}
 for all $\mb Z\in\mr{null}\left(\mc A\right)$. Applying H\"{o}lder's
inequality to the first two terms shows that it suffices to find a
$\mb Y\in\mr{range}\left(\mc A^{*}\right)$ that satisfies 
\begin{align}
-\left\Vert \widehat{\mb h}\overline{\mb x}^{*}-\mc P_{\msf T}\left(\mb Y\right)\right\Vert _{F}\left\Vert \mc P_{\msf T}\left(\mb Z\right)\right\Vert _{F}+\left(1-\left\Vert \mc P_{\msf T^{\perp}}\left(\mb Y\right)\right\Vert \right)\left\Vert \mc P_{\msf T^{\perp}}\left(\mb Z\right)\right\Vert _{*} & >0,\label{eq:UniquenessSufficient}
\end{align}
 for all $\mb Z\in\mr{null}\left(\mc A\right)$. Using the fact that
$\mb Z\in\mr{null}\left(\mc A\right)$ and Lemma \ref{lem:ANearIsometry} below
which guarantees
\begin{align}
\left|\left\Vert \mc A\left(\mb X\right)\right\Vert _{F}^{2}-\mbb E\left[\left\Vert \mc A\left(\mb X\right)\right\Vert _{F}^{2}\right]\right| & \leq\frac{1}{2}\left\Vert \mb X\right\Vert _{F}^{2}\label{eq:ANearIsometry}
\end{align}
 for all $\mb X\in\msf T$, we can deduce that 
\begin{align*}
0=\left\Vert \mc A\left(\mb Z\right)\right\Vert _{F} & \geq\frac{1}{\sqrt{2}}\left\Vert \mc P_{\msf T}\left(\mb Z\right)\right\Vert _{F}-\left\Vert \mc A\right\Vert \left\Vert \mc P_{\msf T^{\perp}}\left(\mb Z\right)\right\Vert _{*}
\end{align*}
 holds with probability at least $1-3L^{-\beta}$ if $K\overset{\beta}{\gtrsim}\mu\log^{2}L\,\log\log\left(N+1\right)$
for some $\beta>0$. The bound (\ref{eq:ANearIsometry}) also guarantees
that $\mc P_{\msf T^{\perp}}\left(\mb Z\right)=\mb Z-\mc P_{\msf T}\left(\mb Z\right)$
cannot be the zero matrix. Combining these results and (\ref{eq:UniquenessSufficient})
shows that it suffices to find a $\mb Y\in\mr{range}\left(\mc A^{*}\right)$
(i.e., the dual certificate) that obeys 
\begin{align}
\sqrt{2}\left\Vert \mc A\right\Vert \left\Vert \mc P_{\msf T}\left(\mb Y\right)-\widehat{\mb h}\overline{\mb x}^{*}\right\Vert _{F} & \leq\frac{1}{4}\label{eq:OnT}
\end{align}
 and 
\begin{align}
\left\Vert \mc P_{\msf T^{\perp}}\left(\mb Y\right)\right\Vert  & <\frac{3}{4}.\label{eq:OnTperp}
\end{align}
\emph{ }

Similar to \cite{ahmed_blind_2014} we employ the golfing scheme \cite{gross-recovering-2011}
to construct the dual certificate $\mb Y$. Consider a partition
of the the index set $\left\{ 1,2,\dotsc,K\right\} $ to its disjoint
subsets $\msf K_{1},\msf K_{2},\dotsc$, and $\msf K_{P}$ such that
\begin{align*}
\left|\msf K_{p}\right| & =\frac{K}{P}\quad\forall p\in\left\{ 1,2,\dotsc,P\right\} .
\end{align*}
 Define the operator restricted to indices in $\msf K_{p}$ as 
\begin{align}
\mc A_{p}\left(\mb X\right) & :=\sqrt{\frac{LP}{K}}\left(\mb F_{N}\odot\mb X\right)\mb{\Phi}_{\msf K_{p}}.\label{eq:Op_p}
\end{align}
 Furthermore, we obtain a sequence of matrices $\mb Y_{0}=\mb 0,\mb Y_{1},\dotsc,\mb Y_{P}$
through the recursive relation 
\begin{align*}
\mb Y_{p} & =\mb Y_{p-1}+\mc A_{p}^{*}\mc A_{p}\left(\widehat{\mb h}\overline{\mb x}^{*}-\mc P_{\msf T}\left(\mb Y_{p-1}\right)\right).
\end{align*}
 Our goal is to show that $\mb Y=\mb Y_{P}$ satisfies (\ref{eq:OnT})
and (\ref{eq:OnTperp}) with high probability. With 
\begin{align}
\mb W_{p} & :=\mc P_{\msf T}\left(\mb Y_{p}\right)-\widehat{\mb h}\overline{\mb x}^{*},\label{eq:Wp}
\end{align}
projecting both sides of the above recursion onto $\msf T$ yields
\begin{align*}
\mb W_{p} & =\mb W_{p-1}-\mc P_{\msf T}\mc A_{p}^{*}\mc A_{p}\left(\mb W_{p-1}\right)\\
 & =\left(\mc P_{\msf T}-\mc P_{\msf T}\mc A_{p}^{*}\mc A_{p}\mc P_{\msf T}\right)\left(\mb W_{p-1}\right),
\end{align*}
 where the latter equation holds because $\mb W_{p-1}\in\msf T$ by
construction. Therefore, with 
\begin{align*}
\left|\msf K_{p}\right|=\frac{K}{P} & \overset{\beta}{\gtrsim}\mu\log^{2}L\,\log\log\left(N+1\right),
\end{align*}
we can invoke Lemma \ref{lem:ANearIsometry} below to guarantee that 
\begin{align*}
\left\Vert \mb W_{p}\right\Vert _{F} & \leq\frac{1}{2}\left\Vert \mb W_{p-1}\right\Vert _{F}\quad\forall p=1,2,\dotsc P,
\end{align*}
 and thus 
\begin{align}
\left\Vert \mb W_{p}\right\Vert _{F} & \leq2^{-p}\left\Vert \widehat{\mb h}\overline{\mb x}^{*}\right\Vert _{F}=2^{-p}\quad\forall p=1,2,\dotsc,P,\label{eq:WpFrobenius}
\end{align}
 hold with probability at least $1-3PL^{-\beta}$. This result implies
that (\ref{eq:OnT}) holds for $\mb Y=\mb Y_{P}$ if 
\begin{align*}
P & \geq\log_{2}\left(4\sqrt{2}\left\Vert \mc A\right\Vert \right)=\frac{5}{2}+\log_{2}\left\Vert \mc A\right\Vert .
\end{align*}
 From Lemma \ref{lem:normA} below we know that that $\left\Vert \mc A\right\Vert \lesssim1+\sqrt{\frac{L}{K}}+\sqrt{\frac{\beta\log L}{K}}$
with probability at least $1-L^{-\beta}$. Therefore, we can deduce
that with 
\begin{align}
P & \gtrsim\max\left\{ 1,\log\left(1+\sqrt{\frac{L}{K}}+\sqrt{\frac{\beta\log L}{K}}\right)\right\} ,\label{eq:P-lowerbound}
\end{align}
(\ref{eq:OnT}) holds for $\mb Y=\mb Y_{P}$ with probability exceeding
$1-\left(3P+1\right)L^{-\beta}$. 

To show that $\mb Y=\mb Y_{P}$ also obeys (\ref{eq:OnTperp}), we
begin by expressing each $\mb Y_{P}$ explicitly in terms of the matrices
$\mb W_{p}$ as 
\begin{align*}
\mb Y_{P} & =\sum_{p=1}^{P}\mb Y_{p}-\mb Y_{p-1}=-\sum_{p=1}^{P}\mc A_{p}^{*}\mc A_{p}\left(\mb W_{p-1}\right).
\end{align*}
Then using the fact that $\mc P_{\msf T^{\perp}}\left(\mb W_{p-1}\right)=\mb 0$
for all $p=1,2,\dotsc,P$ we can write 
\begin{align*}
\left\Vert \mc P_{\msf T^{\perp}}\left(\mb Y_{P}\right)\right\Vert  & =\left\Vert \mc P_{\msf T^{\perp}}\left(\sum_{p=1}^{P}\mc A_{p}^{*}\mc A_{p}\left(\mb W_{p-1}\right)\right)\right\Vert \\
 & =\left\Vert \mc P_{\msf T^{\perp}}\left(\sum_{p=1}^{P}\mc A_{p}^{*}\mc A_{p}\left(\mb W_{p-1}\right)-\mb W_{p-1}\right)\right\Vert \\
 & \leq\left\Vert \sum_{p=1}^{P}\mc A_{p}^{*}\mc A_{p}\left(\mb W_{p-1}\right)-\mb W_{p-1}\right\Vert \\
 & \leq\sum_{p=1}^{P}\left\Vert \mc A_{p}^{*}\mc A_{p}\left(\mb W_{p-1}\right)-\mb W_{p-1}\right\Vert .
\end{align*}
If the size of each partition $\msf K_{p}$ is sufficiently large
and specifically obeys 
\begin{align}
\frac{K}{P}=\left|\msf K_{p}\right| & \overset{\beta}{\gtrsim}\mu\log^{2}L\log\log\left(N+1\right),\label{eq:P-upperbound}
\end{align}
then we can apply Lemma \ref{lem:Ap*Ap}, stated below in Section \ref{sub:Decay}, to simplify the bound on
$\left\Vert \mc P_{\msf T^{\perp}}\left(\mb Y_{P}\right)\right\Vert $
and write 
\begin{align*}
\left\Vert \mc P_{\msf T^{\perp}}\left(\mb Y_{P}\right)\right\Vert  & \leq\frac{3}{4}\sum_{p=1}^{P}2^{-p}<\frac{3}{4},
\end{align*}
which holds with probability at least $1-cPL^{-\beta}$ where $c$
is an absolute constant. Therefore, if there exists a $P$ that satisfies
both (\ref{eq:P-lowerbound}) and (\ref{eq:P-upperbound}), then (\ref{eq:OnT})
and (\ref{eq:OnTperp}) simultaneously hold for $\mb Y=\mb Y_{P}$
with probability at least $1-c'PL^{-\beta}$ for some absolute constant
$c'$. To guarantee existence of such $P$, it suffices to have 
\begin{align*}
K & \overset{\beta}{\gtrsim}\mu\log^{2}L\log\frac{Le}{\mu}\log\log\left(N+1\right),
\end{align*}
 for which we can choose $P\lesssim\log\frac{Le}{\mu}$. Since $\mu\geq L/N$
we have $P\lesssim N$. The probability of the desired events exceeds
$1-O\left(NL^{-\beta}\right)$. 
\end{proof}

\subsubsection{$\protect\mc A_{p}$ is a near isometry on $\protect\msf T$}

We would like to show that the restriction of $\mc A_{p}$, defined
by (\ref{eq:Op_p}), to the subspace $\msf{T}$ in (\ref{eq:Support})
has a near isometry behavior. In particular, our goal is to show that
for all $\mb X\in\msf T$ the inequality 
\begin{align*}
\left|\left\Vert \mc A_{p}\left(\mb X\right)\right\Vert _{F}^{2}-\mbb E\left[\left\Vert \mc A_{p}\left(\mb X\right)\right\Vert _{F}^{2}\right]\right| & \leq\frac{1}{2}\left\Vert \mb X\right\Vert _{F}^{2}
\end{align*}
holds with high probability. The following lemma with $\msf K=\msf K_{p}$
establishes the desired property.
\begin{lem}[near isometry of $\mc A_{\msf K}$ on $\msf T$]
\label{lem:ANearIsometry}Let be $\msf K\subseteq\left\{ 1,2,\dotsc,K\right\} $
be an arbitrary index set and define 
\begin{align}
\mc A_{\msf K}\left(\mb X\right) & :=\sqrt{\frac{L}{\left|\msf K\right|}}\left(\mb F_{N}\odot\mb X\right)\mb{\Phi}_{\msf K}.\label{eq:AK}
\end{align}
For any $\beta>0$, if we have 
\begin{align*}
\left|\msf K\right| & \overset{\beta}{\gtrsim}\mu\log^{2}L\,\log\log\left(N+1\right),
\end{align*}
 then
\begin{align*}
\left|\left\Vert \mc A_{\msf K}\left(\mb X\right)\right\Vert _{F}^{2}-\mbb E\left[\left\Vert \mc A_{\msf K}\left(\mb X\right)\right\Vert _{F}^{2}\right]\right| & \leq\frac{1}{2}\left\Vert \mb X\right\Vert _{F}^{2}
\end{align*}
 for all $\mb X\in\msf T$, with probability at least $1-3L^{-\beta}$.\end{lem}
\begin{proof}
For every $\mb X=\widehat{\mb h}\mb v^{*}+\mb u\overline{\mb x}^{*}\in\msf{T}$ we have
\begin{align}
\left\Vert \mc A_{\msf K}\left(\mb X\right)\right\Vert _{F}^{2}-\mbb E\left[\left\Vert \mc A_{\msf K}\left(\mb X\right)\right\Vert _{F}^{2}\right] & =\frac{L}{\left|\msf K\right|}\sum_{k\in\msf K}\mr{tr}\left(\left(\mb F{}_{N}\odot\mb X\right)\left(\mb{\phi}_{k}\mb{\phi}_{k}^{*}-\mb I\right)\left(\mb F{}_{N}\odot\mb X\right)^{*}\right)\nonumber \\
 & =\frac{L}{\left|\msf K\right|}\sum_{k\in\msf K}\mr{tr}\left.\Bigl(\left(\mb D_{\widehat{\mb h}}\mb F_{N}\mb D_{\mb v}^{*}+\mb D_{\mb u}\mb F_{N}\mb D_{\mb x}\right)\left(\mb{\phi}_{k}\mb{\phi}_{k}^{*}-\mb I\right)\right.\nonumber \\
 &\hspace{5em}\left.\left(\mb D_{\mb v}\mb F_{N}^{*}\mb D_{\widehat{\mb h}}^{*}+\mb D_{\mb x}^{*}\mb F_{N}^{*}\mb D_{\mb u}^{*}\right)\right)\nonumber \\
 & =\frac{L}{\left|\msf K\right|}\sum_{k\in\msf K}\mr{tr}\left(\mb D_{\widehat{\mb h}}\mb F_{N}\mb D_{\mb v}^{*}\left(\mb{\phi}_{k}\mb{\phi}_{k}^{*}-\mb I\right)\mb D_{\mb v}\mb F_{N}^{*}\mb D_{\widehat{\mb h}}^{*}\right)\nonumber \\
 & +\frac{L}{\left|\msf K\right|}\sum_{k\in\msf K}\mr{tr}\left(\mb D_{\mb u}\mb F_{N}\mb D_{\mb x}\left(\mb{\phi}_{k}\mb{\phi}_{k}^{*}-\mb I\right)\mb D_{\mb x}^{*}\mb F_{N}^{*}\mb D_{\mb u}^{*}\right)\nonumber \\
 & +\frac{L}{\left|\msf K\right|}\sum_{k\in\msf K}\mr{tr}\left(\mb D_{\widehat{\mb h}}\mb F_{N}\mb D_{\mb v}^{*}\left(\mb{\phi}_{k}\mb{\phi}_{k}^{*}-\mb I\right)\mb D_{\mb x}^{*}\mb F_{N}^{*}\mb D_{\mb u}^{*}\right)\nonumber \\
 & +\frac{L}{\left|\msf K\right|}\sum_{k\in\msf K}\mr{tr}\left(\mb D_{\mb u}\mb F_{N}\mb D_{\mb x}\left(\mb{\phi}_{k}\mb{\phi}_{k}^{*}-\mb I\right)\mb D_{\mb v}\mb F_{N}^{*}\mb D_{\widehat{\mb h}}^{*}\right).\label{eq:BDDevExpansion}
\end{align}
Therefore, for all $\mb{X}\in\msf{T}$ we can write (\ref{eq:BDDevExpansion}) as 
\begin{align*}
\left\Vert \mc A_{\msf K}\left(\mb X\right)\right\Vert _{F}^{2}-\mbb E\left[\left\Vert \mc A_{\msf K}\left(\mb X\right)\right\Vert _{F}^{2}\right] & =\frac{1}{\left|\msf K\right|}\sum_{k\in\msf K}\left\langle \mb Z_{k},\overline{\mb v}\overline{\mb v}^{*}\right\rangle +\frac{1}{\left|\msf K\right|}\sum_{k\in\msf K}\left\langle \mb Z'_{k},\overline{\mb u}\overline{\mb u}^{*}\right\rangle +\frac{2}{\left|\msf K\right|}\sum_{k\in\msf K}\Re\left\langle \mb Z''_{k},\overline{\mb v}\mb u^{*}\right\rangle .
\end{align*}
where the summands are expressed using the matrices 
\begin{align}
\mb Z_{k} & =L\left(\mb D_{\mb{\phi}_{k}}^{*}\mb F_{N}^{*}\mb D_{\left|\widehat{\mb h}\right|^{2}}\mb F_{N}\mb D_{\mb{\phi}_{k}}-\frac{\left\Vert \widehat{\mb h}\right\Vert _{2}^{2}}{L}\mb I\right),\label{eq:Z_k}
\end{align}
 
\begin{align}
\mb Z'_{k} & =L\mr{diag}\left(\mb F_{N}\mb D_{\mb x}\left(\mb{\phi}_{k}\mb{\phi}_{k}^{*}-\mb I\right)\mb D_{\mb x}^{*}\mb F_{N}^{*}\right),\label{eq:Z'_k}
\end{align}
 and 
\begin{align}
\mb Z''_{k} & =L\left(\mb D_{\mb{\phi}_{k}}^{*}\mb F_{N}^{*}\mb D_{\widehat{\mb h}}^{*}\mb D_{\mb F_{N}\left(\mb x\odot\mb{\phi}_{k}\right)}-\frac{1}{L}\mb x\widehat{\mb h}^{*}\right)\label{eq:Z''_k}
\end{align}
Then, the triangle inequality yields 
\begin{align*}
\left|\left\Vert \mc A_{\msf K}\left(\mb X\right)\right\Vert _{F}^{2}-\mbb E\left[\left\Vert \mc A_{\msf K}\left(\mb X\right)\right\Vert _{F}^{2}\right]\right| & \leq\frac{\left\Vert \mb v\right\Vert _{2}^{2}}{\left|\msf K\right|}\left\Vert \sum_{k\in\msf K}\mb Z{}_{k}\right\Vert+\frac{\left\Vert \mb u\right\Vert _{2}^{2}}{\left|\msf K\right|}\left\Vert \sum_{k\in\msf K}\mb Z'_{k}\right\Vert +\frac{2\left\Vert \mb u\right\Vert _{2}\left\Vert \mb v\right\Vert _{2}}{\left|\msf K\right|}\left\Vert \sum_{k\in\msf K}\mb Z''_{k}\right\Vert.
\end{align*}
Without loss of generality we can assume that $\mb{\overline{v}}$ is orthogonal
to $\mb x$ which implies that 
\begin{align*}
\left\Vert \mb X\right\Vert _{F}^{2} & =\left\Vert \widehat{\mb h}\right\Vert _{2}^{2}\left\Vert \mb v\right\Vert _{2}^{2}+\left\Vert \mb u\right\Vert _{2}^{2}\left\Vert \mb x\right\Vert _{2}^{2}=\left\Vert \mb u\right\Vert _{2}^{2}+\left\Vert \mb v\right\Vert _{2}^{2},
\end{align*}
 and thereby 
\begin{align}
\left|\left\Vert \mc A_{\msf K}\left(\mb X\right)\right\Vert _{F}^{2}-\mbb E\left[\left\Vert \mc A_{\msf K}\left(\mb X\right)\right\Vert _{F}^{2}\right]\right| & \leq\frac{1}{\left|\msf K\right|}\left(\left\Vert \sum_{k\in\msf K}\mb Z{}_{k}\right\Vert +\left\Vert \sum_{k\in\msf K}\mb Z'_{k}\right\Vert +\left\Vert \sum_{k\in\msf K}\mb Z''_{k}\right\Vert \right)\left\Vert \mb X\right\Vert _{F}^{2},\label{eq:BDMainDevBound}
\end{align}
holds for all $\mb{X}\in\msf{T}$.
Therefore, it suffices to bound the operator norm of the sums
of $\mb Z_{k}$s, $\mb Z'_{k}$s, and $\mb Z''_{k}$s, separately.
As shown by Lemmas \ref{lem:HH}, \ref{lem:XX}, and \ref{lem:XH} in Section \ref{ssec:Lemmas}, the
matrix Bernstein's inequality can be used to establish the desired
bounds. It follows from these lemmas that for $\left|\msf K\right|\overset{\beta}{\gtrsim}\mu\log^{2}L\,\log\log\left(N+1\right)$,
the bound 
\begin{align*}
\left|\left\Vert \mc A_{\msf K}\left(\mb X\right)\right\Vert _{F}^{2}-\mbb E\left[\left\Vert \mc A_{\msf K}\left(\mb X\right)\right\Vert _{F}^{2}\right]\right| & \leq\frac{1}{2}\left\Vert \mb X\right\Vert _{F}^{2}
\end{align*}
holds for all $\mb{X}\in\msf{T}$ with probability at least $1-3L^{-\beta}$. 
\end{proof}

\subsubsection{Operator norm of $\protect\mc A$} The following lemma establishes a global bound for the operator norm of $\mc A$ that holds with high probability.
\begin{lem}[the operator norm of $\mc A$]
\label{lem:normA}For any $\beta>0$, the operator norm of $\mc A$
can be bounded as 
\begin{align*}
\left\Vert \mc A\right\Vert  & \lesssim\sqrt{\frac{L}{K}}+\sqrt{\frac{\beta\log L}{K}}
\end{align*}
 with probability at least $1-L^{-\beta}$.\end{lem}
\begin{proof}
By definition we have 
\begin{align*}
\left\Vert \mc A\right\Vert  & =\sup_{\mb X\neq\mb 0}\frac{\left\Vert \mc A\left(\mb X\right)\right\Vert _{F}}{\left\Vert \mb X\right\Vert _{F}}\\
 & =\sup_{\mb X\neq\mb 0}\frac{\left\Vert \sqrt{\frac{L}{K}}\left(\mb F_{N}\odot\mb X\right)\mb{\Phi}\right\Vert _{F}}{\left\Vert \mb X\right\Vert _{F}}\\
 & \leq\sqrt{\frac{L}{K}}\left\Vert \mb{\Phi}\right\Vert \sup_{\mb X\neq\mb 0}\frac{\left\Vert \left(\mb F_{N}\odot\mb X\right)\right\Vert _{F}}{\left\Vert \mb X\right\Vert _{F}}\\
 & =\frac{1}{\sqrt{K}}\left\Vert \mb{\Phi}\right\Vert ,
\end{align*}
where the last inequality holds because $\left\Vert \mb F_{N}\odot\mb X\right\Vert _{F}=\frac{1}{\sqrt{L}}\left\Vert \mb X\right\Vert _{F}$.
Therefore, we can use standard tail bounds for the spectral norm of
random matrices with independent sub-Gaussian entries to bound $\left\Vert \mb{\Phi}\right\Vert $
and thus $\left\Vert \mc A\right\Vert $. For example, the bound established
in \cite[Proposition 2.4]{rudelson_non-asymptotic-2010} guarantees
that for any $t>0$ we have 
\begin{align*}
\left\Vert \mb{\Phi}\right\Vert  & \leq C\left(\sqrt{L}+\sqrt{K}\right)+t
\end{align*}
 with probability at least $1-2e^{-ct^{2}}$, where $C$ and $c$
are absolute constants. Therefore, setting $t=\sqrt{\frac{\beta\log L+\log2}{c}}$
we can show that 
\begin{align*}
\left\Vert \mb{\Phi}\right\Vert  & \lesssim\sqrt{L}+\sqrt{K}+\sqrt{\beta\log L}
\end{align*}
holds with probability at least $1-L^{-\beta}$. This completes the
proof.
\end{proof}

\subsubsection{\label{sub:Decay}Decay of $\left\Vert \protect\mc A_{p}^{*}\protect\mc A_{p}\left(\protect\mb W_{p-1}\right)-\protect\mb W_{p-1}\right\Vert $
with respect to $p$}

Our goal in this section is to show that for $\mc A_{p}$s defined by (\ref{eq:Op_p})
and $\mb W_{p}$s defined by (\ref{eq:Wp}), with high probability, the quantity $\left\Vert \mc A_{p}^{*}\mc A_{p}\left(\mb W_{p-1}\right)-\mb W_{p-1}\right\Vert $
decays quickly as $p$ increases.
\begin{lem}
\label{lem:Ap*Ap}For $p=0,1,\dotsc,P-1$ let $\mb W_{p}$ be defined
by (\ref{eq:Wp}). Then we can choose 
\begin{align*}
\left|\msf K_{p}\right| & \overset{\beta}{\gtrsim}\mu\log^{2}L\log\log\left(N+1\right),
\end{align*}
 such that 
\begin{align*}
\left\Vert \mc A_{p}^{*}\mc A_{p}\left(\mb W_{p-1}\right)-\mb W_{p-1}\right\Vert  & \leq\frac{3}{4}\cdot2^{-p}
\end{align*}
holds for every $p$ simultaneously with probability at least $1-cPL^{-\beta}$
where $c>0$ is an absolute constant. \end{lem}
\begin{proof}
The action of $\mc A_{p}^{*}\mc A_{p}$ on the matrix $\mb{W}_{p-1}$ obeys
\begin{align*}
\mc A_{p}^{*}\mc A_{p}\left(\mb W_{p-1}\right)-\mb W_{p-1} & =L\overline{\mb F}_{N}\odot\left(\left(\mb F_{N}\odot\mb W_{p-1}\right)\left(\frac{1}{\left|\msf K_{p}\right|}\mb{\Phi}_{\msf K_{p}}\mb{\Phi}_{\msf K_{p}}^{*}-\mb I\right)\right)\\
 & =\frac{1}{\left|\msf K_{p}\right|}\sum_{k\in\msf K_{p}}L\overline{\mb F}_{N}\odot\left(\left(\mb F_{N}\odot\mb W_{p-1}\right)\left(\mb{\phi}_{k}\mb{\phi}_{k}^{*}-\mb I\right)\right).
\end{align*}
Since $\mb W_{p-1}\in\msf T$ we can find vectors $\mb u$ and $\mb v$
such that 
\begin{align*}
\mb W_{p-1} & =\widehat{\mb h}\mb v^{*}+\mb u\overline{\mb x}^{*},
\end{align*}
 which implies that 
\begin{align*}
\mc A_{p}^{*}\mc A_{p}\left(\mb W_{p-1}\right)-\mb W_{p-1} & =\frac{1}{\left|\msf K_{p}\right|}\sum_{k\in\msf K_{p}}L\overline{\mb F}_{N}\odot\left(\left(\mb F_{N}\odot\left(\widehat{\mb h}\mb v^{*}+\mb u\overline{\mb x}^{*}\right)\right)\left(\mb{\phi}_{k}\mb{\phi}_{k}^{*}-\mb I\right)\right).
\end{align*}
Without loss of generality, we also assume that $\mb x$ and $\mb{\overline{v}}$
are orthogonal so that 
\begin{align*}
\left\Vert \mb W_{p-1}\right\Vert _{F}^{2} & =\left\Vert \widehat{\mb h}\mb v^{*}\right\Vert _{F}^{2}+\left\Vert \mb u\overline{\mb x}^{*}\right\Vert _{F}^{2}=\left\Vert \mb v\right\Vert _{2}^{2}+\left\Vert \mb u\right\Vert _{2}^{2}.
\end{align*}
Therefore, on the event that the near isometry of $\mc A_{p}$ on
$\msf T$ as stated by the Lemma \ref{lem:ANearIsometry} holds, the
bound in (\ref{eq:WpFrobenius}) guarantees that 
\begin{align}
\left\Vert \mb u\right\Vert _{2}^{2}+\left\Vert \mb v\right\Vert _{2}^{2} & \leq2^{-2p}.\label{eq:||u||-||v||}
\end{align}
 Furthermore, we can write 
\begin{align}
\mc A_{p}^{*}\mc A_{p}\left(\mb W_{p-1}\right)-\mb W_{p-1} & =\frac{1}{\left|\msf K_{p}\right|}\left(\sum_{k\in\msf K_{p}}\widetilde{\mb Z}_{k}+\sum_{k\in\msf K_{p}}\widetilde{\mb Z}'_{k}\right)\label{eq:Ap*Ap-TwoSum}
\end{align}
 where 
\begin{align*}
\widetilde{\mb Z}_{k} & :=L\overline{\mb F}_{N}\odot\left(\left(\mb F_{N}\odot\left(\widehat{\mb h}\mb v^{*}\right)\right)\left(\mb{\phi}_{k}\mb{\phi}_{k}^{*}-\mb I\right)\right)\\
 & =L\mb D_{\widehat{\mb h}}\mb D_{\mb F_{N}\left(\overline{\mb v}\odot\mb{\phi}_{k}\right)}\overline{\mb F}_{N}\mb D_{\phi_{k}}^{*}-\widehat{\mb h}\mb v^{*}
\end{align*}
and 
\begin{align*}
\widetilde{\mb Z}'_{k} & :=L\overline{\mb F}_{N}\odot\left(\left(\mb F_{N}\odot\left(\mb u\overline{\mb x}^{*}\right)\right)\left(\mb{\phi}_{k}\mb{\phi}_{k}^{*}-\mb I\right)\right)\\
 & =L\mb D_{\mb u}\mb D_{\mb F_{N}\left(\mb{x}\odot\mb{\phi}_{k}\right)}\overline{\mb F}_{N}\mb D_{\mb{\phi}_{k}}^{*}-\mb u\overline{\mb x}^{*}.
\end{align*}
These matrices are very similar to the matrix $\mb Z''_{k}$ defined
by (\ref{eq:Z''_k}). In fact, if we consider $\mb Z''_{k}$ to be
a function of $\mb x$ and $\mb{\widehat{\mb h}}$ such
as $\mb Z''_{k}\left(\mb x,\widehat{\mb h}\right)$, then it is easy
to verify that 
\begin{align*}
\widetilde{\mb Z_{k}} & =\left(\mb Z''_{k}\left(\overline{\mb v},\overline{\mb{\widehat{h}}}\right)\right)^{\mr T},
\end{align*}
 and 
\begin{align*}
\widetilde{\mb Z}'_{k} & =\left(\mb Z''_{k}\left(\mb x,\overline{\mb u}\right)\right)^{\mr T}.
\end{align*}
 Therefore, we can readily use Lemma \ref{lem:XH} in Section \ref{ssec:Lemmas} to obtain bounds
for the spectral norms of the sum of $\widetilde{\mb Z}_{k}$s and
the sum of $\widetilde{\mb Z}'_{k}$s. To adapt the result of Lemma
\ref{lem:XH}, it suffices to scale the deviation bounds by the norm
of the vectors $\mb u$ or $\mb v$ and replace the coherence of $\widehat{\mb h}$
by that of $\mb u$, as necessary.

As explained above, we can use Lemma \ref{lem:XH} and (\ref{eq:Ap*Ap-TwoSum})
to obtain 
\begin{align*}
\left\Vert \mc A_{p}^{*}\mc A_{p}\left(\mb W_{p-1}\right)-\mb W_{p-1}\right\Vert  & \leq\frac{1}{\left|\msf K_{p}\right|}\left(\left\Vert \sum_{k\in\msf K_{p}}\widetilde{\mb Z}_{k}\right\Vert +\left\Vert \sum_{k\in\msf K_{p}}\widetilde{\mb Z}'_{k}\right\Vert \right)\\
 & \overset{\beta}{\lesssim}\left\Vert \mb v\right\Vert _{2}\max\left\{ \sqrt{\frac{\mu}{\left|\msf K_{p}\right|}\log L},\frac{\sqrt{\mu}\log L\log\left(N+1\right)\,\log\log\left(N+1\right)}{\left|\msf K_{p}\right|}\right\} \\
 & +\max\left\{ \sqrt{\frac{ L\left\Vert \mb u\right\Vert _{\infty}^{2}}{\left|\msf K_{p}\right|}\log L},\right.\\
 &\hspace{5em}\left.\frac{\sqrt{L\left\Vert \mb u\right\Vert _{\infty}^{2}}\log L\log\left(N+1\right)\,\log\log\left(N+1\right)}{\left|\msf K_{p}\right|}\right\} ,
\end{align*}
with probability at least $1-2L^{-\beta}$. We define $\mu_{p-1}$,
a quantity that controls the largest row-wise energy of $\mb W_{p-1}$,
by 
\begin{align}
\mu_{p-1} & :=L\left\Vert \mb W_{p-1}\right\Vert _{\infty,2}^{2},\label{eq:maxRowWp}
\end{align}
where $\left\Vert\cdot\right\Vert_{\infty,2}$ denotes the largest row-wise $\ell_2$-norm.
Since $\mb u=\mb W_{p-1}\overline{\mb x}$ and $\overline{\mb x}$
is unit-norm, it is straightforward to show that 
\begin{align*}
\left\Vert \mb u\right\Vert _{\infty}^{2} & \leq\frac{\mu_{p-1}}{L},
\end{align*}
and rewrite the bound on $\left\Vert \mc A_{p}^{*}\mc A_{p}\left(\mb W_{p-1}\right)-\mb W_{p-1}\right\Vert $
as 
\begin{align*}
\left\Vert \mc A_{p}^{*}\mc A_{p}\left(\mb W_{p-1}\right)-\mb W_{p-1}\right\Vert  & \overset{\beta}{\lesssim}\left\Vert \mb v\right\Vert _{2}\max\left\{ \sqrt{\frac{\mu}{\left|\msf K_{p}\right|}\log L},\frac{\sqrt{\mu}\log L\log\left(N+1\right)\,\log\log\left(N+1\right)}{\left|\msf K_{p}\right|}\right\} \\
 & +\max\left\{ \sqrt{\frac{ \mu_{p-1} }{\left|\msf K_{p}\right|}\log L},\right. \\
 &\hspace{5em}\left.\frac{\sqrt{\mu_{p-1}}\log L\log\left(N+1\right)\,\log\log\left(N+1\right)}{\left|\msf K_{p}\right|}\Biggl\}\right. .
\end{align*}
The inequality (\ref{eq:||u||-||v||}) implies that $\left\Vert \mb u\right\Vert _{2}\leq2^{-p}$
and $\left\Vert \mb v\right\Vert _{2}\leq2^{-p}$. Furthermore, if
we have 
\begin{align*}
\left|\msf K_{p}\right| & \overset{\beta}{\gtrsim}\mu\log^{2}L\log\log\left(N+1\right),
\end{align*}
then we can invoke Lemma \ref{lem:maxRowWp} below to guarantee the bound
$\mu_{p-1}\leq2^{-2p}\mu$ with probability exceeding $1-2L^{-\beta}$.
Therefore, the above deviation bound can be simplified to 
\begin{align*}
\left\Vert \mc A_{p}^{*}\mc A_{p}\left(\mb W_{p-1}\right)-\mb W_{p-1}\right\Vert  & \overset{\beta}{\lesssim}2^{-p}\max\left\{ \sqrt{\frac{\mu}{\left|\msf K_{p}\right|}\log L},\frac{\sqrt{\mu}\log L\log\left(N+1\right)\,\log\log\left(N+1\right)}{\left|\msf K_{p}\right|}\right\} \\
 & \leq\frac{3}{4}\cdot2^{-p}.
\end{align*}
 The events that the above inequality depends on for every $p=1,2,\dotsc,P$,
hold simultaneously with probability exceeding $1-cPL^{-\beta}$ where
$c>0$ is an absolute constant. 
\end{proof}

\subsubsection{Controlling the largest row-norm of $\protect\mb W_{p}$}

Through the following lemma we show that the largest row-wise $\ell_{2}$-norm
of $\mb W_{p}$ decreases significantly as $p$ increases.
\begin{lem}
\label{lem:maxRowWp}For $p=1,2,\dotsc,P$, let $\mu_{p}$ be defined
as (\ref{eq:maxRowWp}). Furthermore, suppose that 
\begin{align*}
\mu_{p-1} & \leq2^{-2\left(p-1\right)}\mu.
\end{align*}
Then, with $\left|\msf K_{p}\right|\overset{\beta}{\gtrsim}\mu\log^{2}L\,\log\log\left(N+1\right)$
we have 
\begin{align*}
\mu_{p} & \leq2^{-2p}\mu,
\end{align*}
 with probability at least $1-2L^{-\beta}$. Therefore, we have 
\begin{align*}
\mu_{p} & \leq2^{-2p}\mu,
\end{align*}
simultaneously for all $p=1,2,\dotsc,P$ with probability at least
$1-2PL^{-\beta}$.\end{lem}
\begin{proof}
Let $\mb R_{p}:=\mb W_{p-1}-\mc A_{p}^{*}\mc A_{p}\left(\mb W_{p-1}\right)$. Furthermore,
denote the $l$-th columns of $\mb R_{p}^{*}$, $\mb W_{p-1}^{*}$,
and $\mb W_{p}^{*}$ by $\mb r_{l,p}$, $\mb w_{l,p-1}$, and $\mb w_{l,p}$,
respectively. Because $\mb W_{p}=\mc P_{\msf T}\left(\mb R_{p}\right)$,
it follows from Lemma \ref{lem:RowsOnT} in Section \ref{ssec:Lemmas} that 
\begin{align}
\left\Vert \mb w_{l,p}\right\Vert _{2}^{2} & \leq\left\Vert \widehat{\mb h}\right\Vert _{\infty}^{2}\left\Vert \mb R_{p}^{*}\widehat{\mb h}\right\Vert _{2}^{2}+\left|\left\langle \overline{\mb x},\mb r_{l,p}\right\rangle \right|^{2}\nonumber \\
 & =\frac{\mu}{L}\left\Vert \mb R_{p}^{*}\widehat{\mb h}\right\Vert _{2}^{2}+\left|\left\langle \overline{\mb x},\mb r_{l,p}\right\rangle \right|^{2}.\label{eq:w_l,p}
\end{align}
 We can expand $\mb R_{p}$ as 
\begin{align*}
\mb R_{p} & =\frac{1}{\left|\msf K_{p}\right|}\sum_{k\in\msf K_{p}}\mb W_{p-1}-L\overline{\mb F}_{N}\odot\left(\left(\mb F_{N}\odot\mb W_{p-1}\right)\mb{\phi}_{k}\mb{\phi}_{k}^{*}\right)\\
 & =\frac{1}{\left|\msf K_{p}\right|}\sum_{k\in\msf K_{p}}\mb W_{p-1}-L\overline{\mb F}_{N}\odot\left(\left(\mb F_{N}\odot\mb W_{p-1}\right)\mb{\phi}_{k}\mb{\phi}_{k}^{*}\right).
\end{align*}
 Therefore, we can write 
\begin{align*}
\mb R_{p}^{*}\widehat{\mb h} & =\frac{1}{\left|\msf K_{p}\right|}\sum_{k\in\msf K_{p}}\left(\mb W_{p-1}^{*}-L\overline{\mb F}_{N}^{*}\odot\left(\mb{\phi}_{k}\mb{\phi}_{k}^{*}\left(\mb F_{N}^{*}\odot\mb W_{p-1}^{*}\right)\right)\right)\widehat{\mb h}\\
 & =\frac{1}{\left|\msf K_{p}\right|}\sum_{k\in\msf K_{p}}\mb W_{p-1}^{*}\widehat{\mb h}-L\mb D_{\mb{\phi}_{k}}\overline{\mb F}_{N}^{*}\left(\overline{\left(\mb F_{N}\odot\mb W_{p-1}\right)\mb{\phi}_{k}}\odot\widehat{\mb h}\right)\\
 & =\frac{1}{\left|\msf K_{p}\right|}\sum_{k\in\msf K_{p}}\mb W_{p-1}^{*}\widehat{\mb h}-L\mb D_{\mb{\phi}_{k}}\overline{\mb F}_{N}^{*}\left(\overline{\left(\mb F_{N}\odot\mb W_{p-1}\right)\mb{\phi}_{k}}\odot\widehat{\mb h}\right)\\
 & =\frac{1}{\left|\msf K_{p}\right|}\sum_{k\in\msf K_{p}}\mb z_{k},
\end{align*}
where 
\begin{align}
\mb z_{k} & :=\mb W_{p-1}^{*}\widehat{\mb h}-L\mb D_{\mb{\phi}_{k}}\overline{\mb F}_{N}^{*}\left(\overline{\left(\mb F_{N}\odot\mb W_{p-1}\right)\mb{\phi}_{k}}\odot\widehat{\mb h}\right)\label{eq:z_k}
\end{align}
Conditioned on $\mc A_{1},\mc A_{2},\dotsc,$ and $\mc A_{p-1}$,
we can invoke Lemma \ref{lem:z_k} below and show that 
\begin{align}
\left\Vert \mb R_{p}^{*}\widehat{\mb h}\right\Vert _{2}^{2} & \overset{\beta}{\lesssim}2^{-2\left(p-1\right)}\max\left\{ \frac{\mu\log L}{\left|\msf K_{p}\right|},\frac{\mu\log^{2}\left(N+1\right)\log^{2}\log\left(N+1\right)\log^{2}L}{\left|\msf K_{p}\right|^{2}}\right\} \label{eq:mu_p(1)}
\end{align}
holds with probability exceeding $1-L^{-\beta}$. Furthermore, we
can expand $\left\langle \overline{\mb x},\mb r_{l,p}\right\rangle $
as 
\begin{align*}
\left\langle \overline{\mb x},\mb r_{l,p}\right\rangle  & =\frac{1}{\left|\msf K_{p}\right|}\sum_{k\in\msf K_{p}}\left\langle \overline{\mb x},\mb w_{l,p-1}\right\rangle -L\left\langle \overline{\mb x},\overline{\mb f_{l}}\odot\mb{\phi}_{k}\right\rangle \left\langle \overline{\mb f_{l}}\odot\mb{\phi}_{k},\mb w_{l,p-1}\right\rangle \\
 & =\frac{1}{\left|\msf K_{p}\right|}\sum_{k\in\msf K_{p}}\left\langle \overline{\mb x},\mb w_{l,p-1}\right\rangle -L\left\langle \mb f_{l}\odot\overline{\mb x},\mb{\phi}_{k}\right\rangle \left\langle \mb{\phi}_{k},\mb f_{l}\odot\mb w_{l,p-1}\right\rangle \\
 & =-\frac{1}{\left|\msf K_{p}\right|}\sum_{k\in\msf K_{p}}\zeta_{k},
\end{align*}
 with 
\begin{align*}
\zeta_{k} & :=L\left\langle \mb f_{l}\odot\overline{\mb x},\mb{\phi}_{k}\right\rangle \left\langle \mb{\phi}_{k},\mb f_{l}\odot\mb w_{l,p-1}\right\rangle -\left\langle \overline{\mb x},\mb w_{l,p-1}\right\rangle .
\end{align*}
 The Orlicz 1-norm of $\zeta_{k}$ can be bounded using Lemma \ref{lem:RademacherSubExp1} below
as 
\begin{align*}
\left\Vert \zeta_{k}\right\Vert _{\psi_{1}} & \lesssim\left\Vert \mb w_{l,p-1}\right\Vert _{2}\left\Vert \mb x\right\Vert _{2}=\left\Vert \mb w_{l,p-1}\right\Vert _{2}.
\end{align*}
 Furthermore, we have 
\begin{align*}
\mbb E\left[\sum_{k\in\msf K_{p}}\left|\zeta_{k}\right|^{2}\right] & =\left|\msf K_{p}\right|\left(L^{2}\mbb E\left[\left|\left\langle \mb f_{l}\odot\overline{\mb x},\mb{\phi}_{k}\right\rangle \right|^{2}\left|\left\langle \mb{\phi}_{k},\mb f_{l}\odot\mb w_{l,p-1}\right\rangle \right|^{2}\right]-\left|\left\langle \overline{\mb x},\mb w_{l,p-1}\right\rangle \right|^{2}\right)\\
 & \leq3\left|\msf K_{p}\right|\left\Vert \mb w_{l,p-1}\right\Vert _{2}^{2},
\end{align*}
where Lemma \ref{lem:RademacherFoutrhMoment} in Section \ref{ssec:Lemmas} is used to obtain the
inequality. Therefore, the scalar Bernstein's inequality guarantees
that 
\begin{align*}
\left|\sum_{k\in\msf K_{p}}\zeta_{k}\right| & \lesssim\max\left\{ \left\Vert \mb w_{l,p-1}\right\Vert _{2}\sqrt{3\left|\msf K_{p}\right|\left(t+\log2\right)},\left\Vert \mb w_{l,p-1}\right\Vert _{2}\log\left(t+\log2\right)\right\} 
\end{align*}
 with probability at least $1-e^{-t}$. With $t=\beta\log L+\log N$,
we deduce that 
\begin{align}
\left|\left\langle \overline{\mb x},\mb r_{l,p}\right\rangle \right|^{2} & =\left(\frac{1}{\left|\msf K_{p}\right|}\left|\sum_{k\in\msf K_{p}}\zeta_{k}\right|\right)^{2}\nonumber \\
 & \overset{\beta}{\lesssim}\left\Vert \mb w_{l,p-1}\right\Vert _{2}^{2}\max\left\{ \frac{\log L}{\left|\msf K_{p}\right|},\frac{\log^{2}L}{\left|\msf K_{p}\right|^{2}}\right\} ,\label{eq:mu_p(2)}
\end{align}
 holds with probability exceeding $1-N^{-1}L^{-\beta}$.

The inequalities (\ref{eq:w_l,p}), (\ref{eq:mu_p(1)}), and (\ref{eq:mu_p(2)}),
and the assumption 
\begin{align*}
L\left\Vert \mb w_{l,p-1}\right\Vert _{2}^{2} & \leq\mu_{p-1}\leq\mu2^{-2\left(p-1\right)}
\end{align*}
show that 
\begin{align*}
L\left\Vert \mb w_{l,p}\right\Vert _{2}^{2}\overset{\beta}{\lesssim} & 2^{-2\left(p-1\right)}\mu\max\left\{ \frac{\mu\log L}{\left|\msf K_{p}\right|},\frac{\mu\log^{2}\left(N+1\right)\log^{2}\log\left(N+1\right)\log^{2}L}{\left|\msf K_{p}\right|^{2}}\right\} \\
 & +2^{-2\left(p-1\right)}\mu\max\left\{ \frac{\log L}{\left|\msf K_{p}\right|},\frac{\log^{2}L}{\left|\msf K_{p}\right|^{2}}\right\} 
\end{align*}
 holds with probability at least $1-2N^{-1}L^{-\beta}$. Then using
the assumption that $\left|\msf K_{p}\right|\overset{\beta}{\gtrsim}\mu\log^{2}L\,\log\log\left(N+1\right)$
and applying the union bound over $l$ guarantees that 
\begin{align*}
\mu_{p}\leq & 2^{-2p}\mu,
\end{align*}
 with probability exceeding $1-2L^{-\beta}$. Finally, since $\mu_{0}=L\left\Vert \mb W_{0}\right\Vert _{\infty,2}^{2}=L\left\Vert \widehat{\mb h}\right\Vert ^{2}=\mu$,
a recursive application of the above bound guarantees that 
\begin{align*}
\mu_{p} & \leq2^{-2p}\mu,
\end{align*}
 for $p=1,2,\dotsc P$ hold simultaneously with probability exceeding
$1-2PL^{-\beta}$.
\end{proof}

\subsection{\label{ssec:Lemmas}Supplementary lemmas}

The proofs in this section rely on a form of the matrix Bernstein's
inequality borrowed from \cite{koltchinskii_2011_nuclear} and stated
in Proposition \ref{pro:MatrixBernstein's}.
\begin{lem}
\label{lem:HH}For any $\beta>0$ the matrices $\mb Z_{k}$, defined
by (\ref{eq:Z_k}), obey 
\begin{align*}
\left\Vert \sum_{k\in\msf K}\mb Z_{k}\right\Vert  & \overset{\beta}{\lesssim}\max\left\{ \sqrt{\left|\msf K\right|\mu\log L},\mu\log\mu\log L\right\} ,
\end{align*}
with probability at least $1-L^{-\beta}$. \end{lem}
\begin{proof}
Since $\left\Vert \mb Z_{k}\right\Vert \leq L\left\Vert \widehat{\mb h}\right\Vert _{\infty}^{2}=\mu$,
then 
\begin{align*}
\max_{k\in\msf K}\left\Vert \mb Z_{k}\right\Vert _{\psi_{1}} & \leq B:=\frac{\mu}{\log2}.
\end{align*}
 Furthermore, we have 
\begin{align*}
\mbb E\left[\mb Z_{k}^{2}\right] & =L^{2}\left(\frac{\left\Vert \widehat{\mb h}\right\Vert _{4}^{4}}{L}\mb I-\frac{\left\Vert \widehat{\mb h}\right\Vert _{2}^{4}}{L^{2}}\mb I\right)
\end{align*}
from which we obtain the bound 
\begin{align*}
\left\Vert \mbb E\left[\sum_{k\in\msf K}\mb Z_{k}^{2}\right]\right\Vert  & \leq\sigma^{2}:=\left|\msf K\right|\mu.
\end{align*}
 Therefore, for any $t>0$ the matrix Bernstein's inequality guarantees
that 
\begin{align*}
\left\Vert \sum_{k\in\msf K}\mb Z_{k}\right\Vert  & \leq C\max\left\{ \sigma\sqrt{t+\log2L},B\log\left(\frac{\sqrt{\left|\msf K\right|}B}{\sigma}\right)\left(t+\log2L\right)\right\} \\
 & =C\max\left\{ \sqrt{\left|\msf K\right|\mu\left(t+\log2L\right)},\frac{\mu}{\log2}\log\left(\frac{\sqrt{\mu}}{\log2}\right)\left(t+\log2L\right)\right\} 
\end{align*}
 with probability at least $1-e^{-t}$. In particular, with $t=\beta\log L$
we obtain 
\begin{align*}
\left\Vert \sum_{k\in\msf K}\mb Z_{k}\right\Vert  & \overset{\beta}{\lesssim}\max\left\{ \sqrt{\left|\msf K\right|\mu\log L},\mu\log\mu\log L\right\} 
\end{align*}
 with probability at least $1-L^{-\beta}$.\end{proof}
\begin{lem}
\label{lem:XX}For any $\beta>0$ the matrices $\mb Z'_{k}$, as defined
by (\ref{eq:Z'_k}), obey 
\begin{align*}
\left\Vert \sum_{k\in\msf K}\mb Z'_{k}\right\Vert  & \overset{\beta}{\lesssim}\max\left\{ \sqrt{\left|\msf K\right|\log L},\log^{2}L\right\} ,
\end{align*}
with probability at least $1-L^{-\beta}$. \end{lem}
\begin{proof}
Note that $\mb Z'_{k}=L\left(\mb D_{\left|\mb F_{N}\left(\mb x\odot\mb{\phi}_{k}\right)\right|^{2}}-\frac{\left\Vert \mb x\right\Vert _{2}^{2}}{L}\mb I\right)$
using which we deduce that 
\begin{align*}
\left\Vert \mb Z'_{k}\right\Vert  & \leq L\max_{l=1,2,\dotsc,N}\left|\left|\left\langle \mb{\phi}_{k},\overline{\mb x}\odot\mb f_{l}\right\rangle \right|^{2}-\frac{\left\Vert \mb x\right\Vert _{2}^{2}}{L}\right|.
\end{align*}
 Therefore, the Orlicz 1-norm of the right-hand side is an upper bound
for that of $\mb Z'_{k}$. Using Proposition \ref{pro:MaxOrlicz},
we can thus show that 
\begin{align*}
\left\Vert \mb Z'_{k}\right\Vert _{\psi_{1}} & \leq c_{\psi_{1}}L\log\left(N+1\right)\max_{l=1,2,\dotsc,N}\left\Vert \left|\left|\left\langle \mb{\phi}_{k},\overline{\mb x}\odot\mb f_{l}\right\rangle \right|^{2}-\frac{\left\Vert \mb x\right\Vert _{2}^{2}}{L}\right|\right\Vert _{\psi_{1}},
\end{align*}
for some absolute constant $c_{\psi_{1}}>0$ depending only on the
function $\psi_{1}\left(u\right)=e^{u}-1$. Applying Lemma \ref{lem:RademacherSubExp1} below
to the latter inequality then yields 
\begin{align*}
\left\Vert \mb Z'_{k}\right\Vert _{\psi_{1}} & \leq8c_{\psi_{1}}\log\left(N+1\right)
\end{align*}
 and thereby
\begin{align*}
\max_{k\in\msf K}\left\Vert \mb Z'_{k}\right\Vert _{\psi_{1}} & \leq B:=8c_{\psi_{1}}\log\left(N+1\right).
\end{align*}
 Furthermore, we have 
\begin{align*}
\mbb E\left[\mb Z_{k}^{\prime2}\right] & =L^{2}\left(\mbb E\left[\mb D_{\left|\mb F_{N}\left(\mb x\odot\mb{\phi}_{k}\right)\right|^{4}}\right]-\frac{\left\Vert \mb x\right\Vert _{2}^{4}}{L^{2}}\mb I\right)\\
 & =L^{2}\left(\mb D_{\mbb E\left[\left|\mb F_{N}\left(\mb x\odot\mb{\phi}_{k}\right)\right|^{4}\right]}-\frac{\left\Vert \mb x\right\Vert _{2}^{4}}{L^{2}}\mb I\right)\\
 & \preccurlyeq2L^{2}\left(\frac{\left\Vert \mb x\right\Vert _{2}^{4}}{L^{2}}-\frac{\left\Vert \mb x\right\Vert _{4}^{4}}{L^{2}}\mb I\right)\\
 & \preccurlyeq2\mb I
\end{align*}
where the first matrix inequality follows from Lemma \ref{lem:RademacherFoutrhMoment} below.
The above inequality then yields 
\begin{align*}
\left\Vert \mbb E\left[\sum_{k\in\msf K}\mb Z_{k}^{\prime2}\right]\right\Vert  & \leq\sigma^{2}:=2\left|\msf K\right|.
\end{align*}
Applying the matrix Bernstein's inequality for $t>0$ shows that that
\begin{align*}
\left\Vert \sum_{k\in\msf K}\mb Z_{k}^{\prime}\right\Vert  & \leq C'\max\left\{ \sigma\sqrt{t+\log2N},B\log\left(\frac{\sqrt{\left|\msf K\right|}B}{\sigma}\right)\left(t+\log2N\right)\right\} \\
 & =C'\max\left\{ \sqrt{2\left|\msf K\right|\left(t+\log2N\right)},8c_{\psi_{1}}\log\left(N+1\right)\log\left(\frac{8c_{\psi_{1}}\log\left(N+1\right)}{\sqrt{2}}\right)\left(t+\log2N\right)\right\} 
\end{align*}
 with probability at least $1-e^{-t}$. Setting $t=\beta\log L$ we
obtain 
\begin{align*}
\left\Vert \sum_{k\in\msf K}\mb Z'_{k}\right\Vert  & \overset{\beta}{\lesssim}\max\left\{ \sqrt{\left|\msf K\right|\log L},\log L\log\left(N+1\right)\log\log\left(N+1\right)\right\} 
\end{align*}
 with probability at least $1-L^{-\beta}$.\end{proof}
\begin{lem}
\label{lem:XH}For any $\beta>0$ the matrices $\mb Z''_{k}$, as
defined by (\ref{eq:Z''_k}), obey 
\begin{align*}
\left\Vert \sum_{k\in\msf K}\mb Z''_{k}\right\Vert  & \overset{\beta}{\lesssim}\max\left\{ \sqrt{\mu\left|\msf K\right|\log L},\sqrt{\mu}\log L\log\left(N+1\right)\,\log\log\left(N+1\right)\right\} ,
\end{align*}
with probability at least $1-L^{-\beta}$. \end{lem}
\begin{proof}
Using the triangle inequality we have 
\begin{align*}
\left\Vert \mb Z''_{k}\right\Vert _{\psi_{1}} & \leq L\left\Vert \mb D_{\mb{\phi}_{k}}^{*}\mb F_{N}^{*}\mb D_{\widehat{\mb h}}^{*}\mb D_{\mb F_{N}\left(\mb x\odot\mb{\phi}_{k}\right)}\right\Vert _{\psi_{1}}+\left\Vert \mb x\widehat{\mb h}^{*}\right\Vert _{\psi_{1}}.
\end{align*}
Straightforward bounds for the spectral norm yield 
\begin{align*}
\left\Vert \mb D_{\mb{\phi}_{k}}^{*}\mb F_{N}^{*}\mb D_{\widehat{\mb h}}^{*}\mb D_{\mb F_{N}\left(\mb x\odot\mb{\phi}_{k}\right)}\right\Vert  & \leq\left\Vert \widehat{\mb h}\right\Vert _{\infty}\left\Vert \mb F_{N}\left(\mb x\odot\mb{\phi}_{k}\right)\right\Vert _{\infty},
\end{align*}
 using which we can write 
\begin{align*}
\left\Vert \mb Z''_{k}\right\Vert _{\psi_{1}} & \leq L\left\Vert \widehat{\mb h}\right\Vert _{\infty}\left\Vert \max_{l=1,2,\dotsc,N}\left|\left\langle \mb{\phi}_{k},\overline{\mb x}\odot\mb f_{l}\right\rangle \right|\right\Vert _{\psi_{1}}+\frac{1}{\log2}\\
 & =\sqrt{L\mu}\left\Vert \max_{l=1,2,\dotsc,N}\left|\left\langle \mb{\phi}_{k},\overline{\mb x}\odot\mb f_{l}\right\rangle \right|\right\Vert _{\psi_{1}}+\frac{1}{\log2}.
\end{align*}
 The Orlicz 1-norm on the right-hand side can be bounded using Proposition
\ref{pro:MaxOrlicz}. Therefore, we have 
\begin{align*}
\left\Vert \max_{l=1,2,\dotsc,N}\left|\left\langle \mb{\phi}_{k},\overline{\mb x}\odot\mb f_{l}\right\rangle \right|\right\Vert _{\psi_{1}} & \leq c_{\psi_{1}}\log\left(N+1\right)\max_{l=1,2,\dotsc,N}\left\Vert \left|\left\langle \mb{\phi}_{k},\overline{\mb x}\odot\mb f_{l}\right\rangle \right|\right\Vert _{\psi_{1}}
\end{align*}
for some absolute constant $c_{\psi_{1}}>0$ that only depends on
the function $\psi_{1}\left(u\right)=e^{u}-1$. Lemma \ref{lem:RademacherSubExp2} below
guarantees that 
\begin{align*}
\left\Vert \left|\left\langle \mb{\phi}_{k},\overline{\mb x}\odot\mb f_{l}\right\rangle \right|\right\Vert _{\psi_{1}} & \leq\frac{8}{\sqrt{L}}.
\end{align*}
Thus, for each $k$ we have 
\begin{align*}
\left\Vert \mb Z''_{k}\right\Vert _{\psi_{1}} & \leq\left(8c_{\psi_{1}}+1\right)\frac{\sqrt{\mu}\log\left(N+1\right)}{\log2},
\end{align*}
 or equivalently 
\begin{align*}
\max_{k\in\mc{\msf K}}\left\Vert \mb Z''_{k}\right\Vert _{\psi_{1}} & \leq B:=\left(8c_{\psi_{1}}+1\right)\frac{\sqrt{\mu}\log\left(N+1\right)}{\log2}.
\end{align*}
 Furthermore, we have 
\begin{align*}
\mbb E\left[\mb Z_{k}^{\prime\prime*}\mb Z_{k}^{\prime\prime}\right] & =L^{2}\mbb E\left[\mb D_{\left|\widehat{\mb h}\odot\mb F_{N}\left(\mb x\odot\mb{\phi}_{k}\right)\right|^{2}}\right]-\widehat{\mb h}\widehat{\mb h}^{*}\\
 & =L\left\Vert \mb x\right\Vert _{2}^{2}\mb D_{\left|\widehat{\mb h}\right|^{2}}-\widehat{\mb h}\widehat{\mb h}^{*}
\end{align*}
which implies that 
\begin{align*}
\left\Vert \mbb E\left[\sum_{k\in\msf K}\mb Z_{k}^{\prime\prime*}\mb Z_{k}^{\prime\prime}\right]\right\Vert  & \leq\left|\msf K\right|L\left\Vert \widehat{\mb h}\right\Vert _{\infty}^{2}=\left|\msf K\right|\mu.
\end{align*}
We also have 
\begin{align*}
\mbb E\left[\mb Z_{k}^{\prime\prime}\mb Z_{k}^{\prime\prime*}\right] & =L^{2}\mbb E\left[\mb D_{\mb{\phi}_{k}}^{*}\mb F_{N}^{*}\mb D_{\left|\widehat{\mb h}\odot\mb F_{N}\left(\mb x\odot\mb{\phi}_{k}\right)\right|^{2}}\mb F_{N}\mb D_{\mb{\phi}_{k}}\right]-\mb x\mb x^{*}
\end{align*}
 which results in 
\begin{align*}
\left\Vert \mbb E\left[\sum_{k\in\msf K}\mb Z_{k}^{\prime\prime}\mb Z_{k}^{\prime\prime*}\right]\right\Vert  & \leq\left|\msf K\right|L^{2}\max_{\mb u:\left\Vert \mb u\right\Vert _{2}=1}\mbb E\left[\sum_{l=1}^{N}\left|\widehat{h}_{l}\right|^{2}\left|\left\langle \mb{\phi}_{k},\mb f_{l}\odot\overline{\mb x}\right\rangle \right|^{2}\left|\left\langle \mb{\phi}_{k},\mb f_{l}\odot\overline{\mb u}\right\rangle \right|^{2}\right]\\
 & \leq\left|\msf K\right|L^{2}\max_{\mb u:\left\Vert \mb u\right\Vert _{2}=1}\sum_{l=1}^{N}\left|\widehat{h}_{l}\right|^{2}\!\left(3\left\Vert \mb f_{l}\odot\mb x\right\Vert _{2}^{2}\left\Vert \mb f_{l}\odot\overline{\mb u}\right\Vert _{2}^{2}-2\left\langle \left|\mb f_{l}\odot\overline{\mb x}\right|^{2}\!,\left|\mb f_{l}\odot\overline{\mb u}\right|^{2}\right\rangle \right)\\
 & \leq\left|\msf K\right|L^{2}\max_{\mb u:\left\Vert \mb u\right\Vert _{2}=1}\sum_{l=1}^{N}3\left|\widehat{h}_{l}\right|^{2}\left\Vert \mb f_{l}\odot\overline{\mb x}\right\Vert _{2}^{2}\left\Vert \mb f_{l}\odot\overline{\mb u}\right\Vert _{2}^{2}\\
 & =3\left|\msf K\right|
\end{align*}
where we used Lemma \ref{lem:RademacherFoutrhMoment} below in the second
inequality. Therefore, we obtain 
\begin{align*}
\max\left\{ \left\Vert \mbb E\left[\sum_{k\in\msf K}\mb Z_{k}^{\prime\prime*}\mb Z_{k}^{\prime\prime}\right]\right\Vert ,\left\Vert \mbb E\left[\sum_{k\in\msf K}\mb Z_{k}^{\prime\prime}\mb Z_{k}^{\prime\prime*}\right]\right\Vert \right\}  & \leq\sigma^{2}:=\max\left\{ \mu,3\right\} \left|\msf K\right|.
\end{align*}
Applying the matrix Bernstein's inequality for $t>0$ shows that 
\begin{align*}
\left\Vert \sum_{k\in\msf K}\mb Z_{k}^{\prime\prime}\right\Vert  & \leq C''\max\left\{ \sigma\sqrt{t+\log\left(N+L\right)},B\log\left(\frac{\sqrt{\left|\msf K\right|}B}{\sigma}\right)\left(t+\log\left(N+L\right)\right)\right\} \\
 & =C''\max\left\{ \sqrt{\left|\msf K\right|\left(t+\log\left(N+L\right)\right)\max\left\{ \mu,3\right\} },\right.\\
 & \left.\left(8c_{\psi_{1}}+1\right)\frac{\sqrt{\mu}\log\left(N+1\right)}{\log2}\log\left(\frac{\left(8c_{\psi_{1}}+1\right)\log\left(N+1\right)}{\log2}\right)\left(t+\log\left(N+L\right)\right)\right\} 
\end{align*}
 with probability at least $1-e^{-t}$. Setting $t=\beta\log L$ we
obtain 
\begin{align*}
\left\Vert \sum_{k\in\msf K}\mb Z''_{k}\right\Vert  & \overset{\beta}{\lesssim}\max\left\{ \sqrt{\left|\msf K\right|\mu\log L},\sqrt{\mu}\log L\log\left(N+1\right)\,\log\log\left(N+1\right)\right\} ,
\end{align*}
 with probability at least $1-L^{-\beta}$. \end{proof}
\begin{lem}
\label{lem:z_k}Let $\mb z_{k}$ and $\mu_{p-1}$ be defined as in
(\ref{eq:z_k}) and (\ref{eq:maxRowWp}), respectively. Furthermore,
as in Lemma \ref{lem:maxRowWp}, suppose that $\mu_{p-1}\leq2^{-2p+2}\mu$.
Then, conditioned on $\mc A_{1},\mc A_{2},\dotsc,$ and $\mc A_{p-1}$,
for $\beta>0$ we have 
\begin{align*}
\left\Vert \sum_{k\in\msf K_{p}}\mb z_{k}\right\Vert _{2}\overset{\beta}{\lesssim} & 2^{-p+1}\max\left\{ \sqrt{\left|\msf K_{p}\right|\mu\log L},\sqrt{\mu}\log\left(N+1\right)\log\log\left(N+1\right)\,\log L\right\} 
\end{align*}
 with probability at least $1-N^{-1}L^{-\beta}$.\end{lem}
\begin{proof}
Applying the triangle inequality to (\ref{eq:z_k}) yields 
\begin{align}
\left\Vert \mb z_{k}\right\Vert _{\psi_{1}} & \leq B:=\left\Vert \mb W_{p-1}^{*}\widehat{\mb h}\right\Vert _{\psi_{1}}+L\left\Vert \mb D_{\mb{\phi}_{k}}\overline{\mb F}_{N}^{*}\left(\overline{\left(\mb F_{N}\odot\mb W_{p-1}\right)\mb{\phi}_{k}}\odot\widehat{\mb h}\right)\right\Vert _{\psi_{1}}.\label{eq:z_k-Orlicz}
\end{align}
 The first term on the right-hand side of (\ref{eq:z_k-Orlicz}) can
be bounded as 
\begin{align*}
\left\Vert \mb W_{p-1}^{*}\widehat{\mb h}\right\Vert _{\psi_{1}} & \leq\frac{\left\Vert \mb W_{p-1}^{*}\widehat{\mb h}\right\Vert _{2}}{\log2}\leq\frac{\left\Vert \mb W_{p-1}^{*}\right\Vert \left\Vert \widehat{\mb h}\right\Vert _{2}}{\log2}\leq\frac{2^{-p+1}}{\log2},
\end{align*}
 where the last inequality follows from the bound $\left\Vert \mb W_{p-1}\right\Vert _{F}\leq2^{-p+1}$
that is established using Lemma \ref{lem:ANearIsometry}. Let $\mb w_{l,p-1}$
denote the $l$-th column of $\mb W_{p-1}^{*}$. The second term on
the right hand side of (\ref{eq:z_k-Orlicz}) can also be bounded
as  
\begin{align*}
L\left\Vert \mb D_{\mb{\phi}_{k}}\overline{\mb F}_{N}^{*}\left(\overline{\left(\mb F_{N}\odot\mb W_{p-1}\right)\mb{\phi}_{k}}\odot\widehat{\mb h}\right)\right\Vert _{\psi_{1}} & =L\left\Vert \overline{\left(\mb F_{N}\odot\mb W_{p-1}\right)\mb{\phi}_{k}}\odot\widehat{\mb h}\right\Vert _{\psi_{1}}=L\left\Vert \mb D_{\left(\mb F_{N}\odot\mb W_{p-1}\right)\mb{\phi}_{k}}^{*}\widehat{\mb h}\right\Vert _{\psi_{1}}\\
 & \leq L\left\Vert \max_{l=1,2,\dotsc,N}\left|\left\langle \mb f_{l}\odot\mb w_{l,p-1},\mb{\phi}_{k}\right\rangle \right|\left\Vert \widehat{\mb h}\right\Vert _{2}\right\Vert _{\psi_{1}}\\
 & \lesssim L\log\left(N+1\right)\max_{l=1,2,\dotsc,N}\left\Vert \left|\left\langle \mb f_{l}\odot\mb w_{l,p-1},\mb{\phi}_{k}\right\rangle \right|\right\Vert _{\psi_{1}}\\
 & \lesssim L\log\left(N+1\right)\max_{l=1,2,\dotsc,N}\left\Vert \mb f_{l}\odot\mb w_{l,p-1}\right\Vert _{2}\\
 & =\sqrt{L}\log\left(N+1\right)\max_{l=1,2,\dotsc,N}\left\Vert \mb w_{l,p-1}\right\Vert _{2}=\sqrt{\mu_{p-1}}\log\left(N+1\right)\\
 & \leq\sqrt{\mu}\log\left(N+1\right)2^{-p+1},
\end{align*}
where the first through sixth lines respectively hold because of the
fact that $\left\Vert \mb D_{\phi_{k}}\overline{\mb F}_{N}^{*}\mb a\right\Vert _{2}=\left\Vert \mb a\right\Vert _{2}$
for any $\mb a\in\mbb C^{N}$, an elementary property of the spectral
norm, Proposition \ref{pro:MaxOrlicz}, Lemma \ref{lem:RademacherSubExp2},
the definition of $\mu_{p-1}$, and the assumed bound on $\mu_{p-1}$.
Therefore, we deduce that $B$ in (\ref{eq:z_k-Orlicz}) obeys 
\begin{align*}
B & \lesssim\sqrt{\mu}\log\left(N+1\right)2^{-p+1}.
\end{align*}
 To apply the matrix Bernstein's inequality we also need to upperbound
the spectral norm of the expectation of the sum of the terms $\mb z_{k}^{*}\mb z_{k}$,
and the sum of the terms $\mb z_{k}\mb z_{k}^{*}$. To bound the first
spectral norm we can write
\begin{align*}
\left\Vert \mbb E\left[\sum_{k\in\msf K_{p}}\mb z_{k}^{*}\mb z_{k}\right]\right\Vert  & =\left|\msf K_{p}\right|\mbb E\left[\left\Vert \mb z_{k}\right\Vert _{2}^{2}\right]\\
 & =\left|\msf K_{p}\right|L^{2}\mbb E\left[\left\Vert \overline{\left(\mb F_{N}\odot\mb W_{p-1}\right)\mb{\phi}_{k}}\odot\widehat{\mb h}\right\Vert _{2}^{2}\right]-\left|\msf K_{p}\right|\left\Vert \mb W_{p-1}^{*}\widehat{\mb h}\right\Vert _{2}^{2}\\
 & =\sum_{l=1}^{N}\left|\msf K_{p}\right|L\left\Vert \mb w_{l,p-1}\right\Vert _{2}^{2}\left|\widehat{h}_{l}\right|^{2}-\left|\msf K_{p}\right|\left\Vert \mb W_{p-1}^{*}\widehat{\mb h}\right\Vert _{2}^{2}\\
 & \leq\left|\msf K_{p}\right|L\left\Vert \widehat{\mb h}\right\Vert _{\infty}^{2}\sum_{l=1}^{N}\left\Vert \mb w_{l,p-1}\right\Vert _{2}^{2}\\
 & =\left|\msf K_{p}\right|\mu\left\Vert \mb W_{p-1}\right\Vert _{F}^{2}\\
 & \leq2^{-2p+2}\left|\msf K_{p}\right|\mu.
\end{align*}
Furthermore, the second spectral norm can be bounded as 
\begin{align*}
\left\Vert \mbb E\left[\sum_{k\in\msf K_{p}}\mb z_{k}\mb z_{k}^{*}\right]\right\Vert  & =\left|\msf K_{p}\right|\left\Vert \mbb E\left[\mb z_{k}\mb z_{k}^{*}\right]\right\Vert \\
 & =\left|\msf K_{p}\right|\max_{\mb u:\left\Vert \mb u\right\Vert _{2}=1}L^{2}\mbb E\left[\left|\left\langle \mb u,\mb z_{k}\right\rangle \right|^{2}\right]\\
 & =\left|\msf K_{p}\right|\max_{\mb u:\left\Vert \mb u\right\Vert _{2}=1}L^{2}\mbb E\left[\left|\sum_{l=1}^{N}\widehat{h}_{l}\left\langle \mb{\phi}_{k},\mb f_{l}\odot\mb u\right\rangle \left\langle \mb f_{l}\odot\mb w_{l,p-1},\mb{\phi}_{k}\right\rangle \right|^{2}\right]-\left|\widehat{\mb h}^{*}\mb W_{p-1}\mb u\right|^{2}\\
 & \leq\left|\msf K_{p}\right|\max_{\mb u:\left\Vert \mb u\right\Vert _{2}=1}L^{2}\mbb E\left[N\sum_{l=1}^{N}\left|\widehat{h}_{l}\right|^{2}\left|\left\langle \mb{\phi}_{k},\mb f_{l}\odot\mb u\right\rangle \right|^{2}\left|\left\langle \mb f_{l}\odot\mb w_{l,p-1},\mb{\phi}_{k}\right\rangle \right|^{2}\right]\\
 & \leq\left|\msf K_{p}\right|\max_{\mb u:\left\Vert \mb u\right\Vert _{2}=1}L^{2}\mbb E\left[N\left\Vert \widehat{\mb h}\right\Vert _{\infty}^{2}\sum_{l=1}^{N}\left|\left\langle \mb{\phi}_{k},\mb f_{l}\odot\mb u\right\rangle \right|^{2}\left|\left\langle \mb f_{l}\odot\mb w_{l,p-1},\mb{\phi}_{k}\right\rangle \right|^{2}\right]\\
 & \leq\left|\msf K_{p}\right|LN\mu\sum_{l=1}^{N}\frac{3\left\Vert \mb w_{l,p-1}\right\Vert _{2}^{2}}{L^{2}}\\
 & =3\left|\msf K_{p}\right|\frac{N}{L}\mu\left\Vert \mb W_{p-1}\right\Vert _{F}^{2}\\
 & \leq3\cdot2^{-2p+2}\left|\msf K_{p}\right|\frac{N}{L}\mu,
\end{align*}
where the first inequality follows from the Cauchy-Schwarz inequality
and the fact that $\left|\widehat{\mb h}^{*}\mb W_{p-1}\mb u\right|^{2}\geq0$,
the second inequality follows from the H\"{o}lder's inequality applied
to the sum inside the expectation, and the third inequality follows
from Lemma \ref{lem:RademacherFoutrhMoment} below. Therefore, we have 
\begin{align*}
\sigma^{2}:=\max\left\{ \left\Vert \mbb E\left[\sum_{k\in\msf K_{p}}\mb z_{k}^{*}\mb z_{k}\right]\right\Vert ,\left\Vert \mbb E\left[\sum_{k\in\msf K_{p}}\mb z_{k}\mb z_{k}^{*}\right]\right\Vert \right\}  & \lesssim\left|\msf K_{p}\right|\mu2^{-2p+2}.
\end{align*}
 Then for any $t>0$ the matrix Bernstein's inequality yields 
\begin{align*}
\left\Vert \sum_{k\in\msf K_{p}}\mb z_{k}\right\Vert _{2} & \leq C\max\left\{ \sigma\sqrt{t+\log\left(L+1\right)},B\log\left(\frac{\sqrt{\left|\msf K_{p}\right|}B}{\sigma}\right)\left(t+\log\left(L+1\right)\right)\right\} \\
 & \lesssim2^{-p+1}\max\left\{ \sqrt{\left|\msf K_{p}\right|\mu\left(t+\log\left(L+1\right)\right)},\right.\\
 & \left.\sqrt{\mu}\log\left(N+1\right)\log\log\left(N+1\right)\,\left(t+\log\left(L+1\right)\right)\right\} .
\end{align*}
with probability at least $1-e^{-t}$. Setting $t=\beta\log L+\log N$
we obtain 
\begin{align*}
\left\Vert \sum_{k\in\msf K_{p}}\mb z_{k}\right\Vert _{2}\overset{\beta}{\lesssim} & 2^{-p+1}\max\left\{ \sqrt{\left|\msf K_{p}\right|\mu\log L},\sqrt{\mu}\log\left(N+1\right)\log\log\left(N+1\right)\,\log L\right\} 
\end{align*}
with probability exceeding $1-N^{-1}L^{-\beta}$.\end{proof}
\begin{lem}
\label{lem:RowsOnT}For any matrix $\mb Z\in\mbb C^{N\times L}$,
the $l$-th row of $\mb Q=\mc P_{\msf T}\left(\mb Z\right)$ denoted
by $\mb q_{l}^{*}$ obeys 
\begin{align*}
\left\Vert \mb q_{l}\right\Vert _{2}^{2} & \leq\left\Vert \widehat{\mb h}\right\Vert _{\infty}^{2}\left\Vert \mb Z^{*}\widehat{\mb h}\right\Vert _{2}^{2}+\left|\left\langle \overline{\mb x},\mb z_{l}\right\rangle \right|^{2},
\end{align*}
 where $\mb z_{l}$ is the $l$-th column of $\mb Z^{*}$.\end{lem}
\begin{proof}
Since $\mb Q$ is the projection of $\mb Z$ onto $\msf T$ we can
write 
\begin{align*}
\mb Q & =\widehat{\mb h}\widehat{\mb h}^{*}\mb Z\left(\mb I-\overline{\mb x}\overline{\mb x}^{*}\right)+\mb Z\overline{\mb x}\overline{\mb x}^{*}.
\end{align*}
 Therefore, $\mb q_{l}^{*}$ (i.e., the $l$-th row of $\mb Q$ )
can be written as 
\begin{align*}
\mb q_{l}^{*} & =\widehat{h}_{l}\widehat{\mb h}^{*}\mb Z\left(\mb I-\overline{\mb x}\overline{\mb x}^{*}\right)+\mb z_{l}^{*}\overline{\mb x}\overline{\mb x}^{*},
\end{align*}
which implies that 
\begin{align*}
\left\Vert \mb q_{l}\right\Vert _{2}^{2} & =\left|\widehat{h}_{l}\right|^{2}\left\Vert \left(\mb I-\overline{\mb x}\overline{\mb x}^{*}\right)\mb Z^{*}\widehat{\mb h}\right\Vert _{2}^{2}+\left|\left\langle \overline{\mb x},\mb z_{l}\right\rangle \right|^{2}\\
 & \leq\left\Vert \widehat{\mb h}\right\Vert _{\infty}^{2}\left\Vert \mb Z^{*}\widehat{\mb h}\right\Vert _{2}^{2}+\left|\left\langle \overline{\mb x},\mb z_{l}\right\rangle \right|^{2}.
\end{align*}
\end{proof}
\begin{lem}
\label{lem:RademacherFoutrhMoment}For any pair of vectors $\mb a$
and $\mb b$, and a Rademacher vector $\mb{\phi}$ we have

\begin{align*}
\mbb E\left[\left|\left\langle \mb a,\mb{\phi}\right\rangle \right|^{2}\left|\left\langle \mb b,\mb{\phi}\right\rangle \right|^{2}\right] & \leq3\left\Vert \mb a\right\Vert _{2}^{2}\left\Vert \mb b\right\Vert _{2}^{2}-2\left\langle \left|\mb a\right|^{2},\left|\mb b\right|^{2}\right\rangle .
\end{align*}
In particular, for $\mb a=\mb b$ we have 
\begin{align*}
\mbb E\left|\left\langle \mb a,\mb{\phi}\right\rangle \right|^{4} & \leq3\left\Vert \mb a\right\Vert _{2}^{4}-2\left\Vert \mb a\right\Vert _{4}^{4}.
\end{align*}
\end{lem}
\begin{proof}
By expanding $\left|\left\langle \mb a,\mb{\phi}\right\rangle \right|^{2}\left|\left\langle \mb b,\mb{\phi}\right\rangle \right|^{2}$
we obtain
\begin{align*}
\mbb E\left[\left|\left\langle \mb a,\mb{\phi}\right\rangle \right|^{2}\left|\left\langle \mb b,\mb{\phi}\right\rangle \right|^{2}\right] & =\mbb E\left[\sum_{i,j,k,l}\phi_{i}\phi_{j}\phi_{k}\phi_{l}a_{i}a_{j}^{*}b_{k}b_{l}^{*}\right]\\
 & =\sum_{i=j,k=l}\left|a_{i}\right|^{2}\left|b_{k}\right|^{2}+\sum_{\substack{\left\{ i,j\right\} =\left\{ k,l\right\} \\
i\neq j
}
}a_{i}a_{j}^{*}b_{k}b_{l}^{*}\\
 & =\left\Vert \mb a\right\Vert _{2}^{2}\left\Vert \mb b\right\Vert _{2}^{2}+\sum_{i\neq j}\left(a_{i}b_{i}\left(a_{j}b_{j}\right)^{*}+a_{i}b_{i}^{*}\left(a_{j}b_{j}^{*}\right)^{*}\right)\\
 & =\left\Vert \mb a\right\Vert _{2}^{2}\left\Vert \mb b\right\Vert _{2}^{2}+\left|\sum_{i}a_{i}b_{i}\right|^{2}+\left|\sum_{i}a_{i}b_{i}^{*}\right|^{2}-2\sum_{i}\left|a_{i}\right|^{2}\left|b_{i}\right|^{2}\\
 & \leq3\left\Vert \mb a\right\Vert _{2}^{2}\left\Vert \mb b\right\Vert _{2}^{2}-2\left\langle \left|\mb a\right|^{2},\left|\mb b\right|^{2}\right\rangle ,
\end{align*}
which is the desired bound.\end{proof}
\begin{lem}
\label{lem:RademacherSubExp1}Let $\mb a$ and $\mb b$ be arbitrary
complex $L$-vectors, and $\mb{\phi}\in\left\{ \pm1\right\} ^{L}$
be a Rademacher vector with independent entries. Then the random variable
$\left\langle \mb a,\mb{\phi}\right\rangle \left\langle \mb{\phi},\mb b\right\rangle -\left\langle \mb a,\mb b\right\rangle $
is subexponential and its Orlicz 1-norm obeys
\begin{align*}
\left\Vert \left\langle \mb a,\mb{\phi}\right\rangle \left\langle \mb{\phi},\mb b\right\rangle -\left\langle \mb a,\mb b\right\rangle \right\Vert _{\psi_{1}} & \lesssim\left\Vert \mb a\right\Vert _{2}\left\Vert \mb b\right\Vert _{2}.
\end{align*}
Specifically, for $\mb a=\mb b$ we have 
\begin{align*}
\left\Vert \left|\left\langle \mb a,\mb{\phi}\right\rangle \right|^{2}-\left\Vert \mb a\right\Vert _{2}^{2}\right\Vert _{\psi_{1}} & \leq8\left\Vert \mb a\right\Vert _{2}^{2}.
\end{align*}
 \end{lem}
\begin{proof}
We begin by proving similar bounds for real vectors $\mb a$. 
\begin{align*}
\mbb E\left[e^{\left|\left\langle \mb a,\mb{\phi}\right\rangle \left\langle \mb{\phi},\mb b\right\rangle -\left\langle \mb a,\mb b\right\rangle \right|/u}\right] & =\int_{0}^{\infty}\mbb P\left(\left|\left\langle \mb a,\mb{\phi}\right\rangle \left\langle \mb{\phi},\mb b\right\rangle -\left\langle \mb a,\mb b\right\rangle \right|>tu\right)e^{t}\mr dt\\
 & \leq2\int_{0}^{\infty}e^{t-c\min\left\{ \left(\frac{tu}{\left\Vert \mb a\right\Vert _{2}\left\Vert \mb b\right\Vert _{2}}\right)^{2},\frac{tu}{\left\Vert \mb a\right\Vert _{2}\left\Vert \mb b\right\Vert _{2}}\right\} }\mr dt
\end{align*}
where the inequality follows from a variant of the Hanson-Wright inequality stated in Proposition \ref{pro:Hanson-Wright}. The latter integral becomes
smaller than one, by choosing $u=C\left\Vert \mb a\right\Vert _{2}\left\Vert \mb b\right\Vert _{2}$
for a sufficiently large absolute constant $C$. Therefore, we can
deduce that 
\begin{align*}
\left\Vert \left\langle \mb a,\mb{\phi}\right\rangle \left\langle \mb{\phi},\mb b\right\rangle -\left\langle \mb a,\mb b\right\rangle \right\Vert _{\psi_{1}} & \leq C\left\Vert \mb a\right\Vert _{2}\left\Vert \mb b\right\Vert _{2}.
\end{align*}

While the above result can be readily used for the special case of
$\mb a=\mb b$, we provide a different proof for this case that does
not rely on the Hanson-Wright inequality. The Hoeffding's inequality
guarantees that 
\begin{align*}
\mbb P\left(\left|\left\langle \mb a,\mb{\phi}\right\rangle \right|>t\right) & \leq2e^{-\frac{t^{2}}{2\left\Vert \mb a\right\Vert _{2}^{2}}},
\end{align*}
 holds for all $t>0$. Therefore, for any $u>0$ we have 
\begin{align*}
\mbb E\left[e^{\frac{\left|\left\langle \mb a,\mb{\phi}\right\rangle \right|^{2}}{u}}\right] & =1+\int_{0}^{\infty}\mbb P\left(e^{\frac{\left|\left\langle \mb a,\mb{\phi}\right\rangle \right|^{2}}{u}}>e^{t^{2}}\right)2te^{t^{2}}\mr dt\\
 & =1+\int_{0}^{\infty}\mbb P\left(\left|\left\langle \mb a,\mb{\phi}\right\rangle \right|>t\sqrt{u}\right)2te^{t^{2}}\mr dt\\
 & \leq1+\int_{0}^{\infty}2e^{t^{2}-\frac{ut^{2}}{2\left\Vert \mb a\right\Vert _{2}^{2}}}\mr dt^{2}\\
 & \leq1+2\int_{0}^{\infty}e^{\left(1-\frac{u}{2\left\Vert \mb a\right\Vert _{2}^{2}}\right)\tau}\mr d\tau.
\end{align*}
In particular, for $u=6\left\Vert \mb a\right\Vert _{2}^{2}$ we obtain
\begin{align*}
\mbb E\left[e^{\frac{\left|\left\langle \mb a,\mb{\phi}\right\rangle \right|^{2}}{4\left\Vert \mb a\right\Vert _{2}^{2}}}\right] & \leq1+\frac{2}{6/2-1}=2,
\end{align*}
which implies $\left\Vert \left|\left\langle \mb a,\mb{\phi}\right\rangle \right|^{2}\right\Vert _{\psi_{1}}\leq6\left\Vert \mb a\right\Vert _{2}^{2}$.
Therefore, using triangle inequality we can deduce that 
\begin{align*}
\left\Vert \left|\left\langle \mb a,\mb{\phi}\right\rangle \right|^{2}-\left\Vert \mb a\right\Vert _{2}^{2}\right\Vert _{\psi_{1}} & \leq\left\Vert \left|\left\langle \mb a,\mb{\phi}\right\rangle \right|^{2}\right\Vert _{\psi_{1}}+\left\Vert \left\Vert \mb a\right\Vert _{2}^{2}\right\Vert _{\psi_{1}}\\
 & \leq6\left\Vert \mb a\right\Vert _{2}^{2}+\frac{1}{\log2}\left\Vert \mb a\right\Vert _{2}^{2}\\
 & \leq8\left\Vert \mb a\right\Vert _{2}^{2}.
\end{align*}

To obtain similar inequalities for complex values of $\mb a$ and
$\mb b$ we can simply decompose the vectors into their real and imaginary
part and apply the triangle inequality. Therefore, we obtain 
\begin{align*}
\left\Vert \left\langle \mb a,\mb{\phi}\right\rangle \left\langle \mb{\phi},\mb b\right\rangle -\left\langle \mb a,\mb b\right\rangle \right\Vert _{\psi_{1}} & \leq\left\Vert \left\langle \Re\mb a,\mb{\phi}\right\rangle \left\langle \mb{\phi},\Re\mb b\right\rangle -\left\langle \Re\mb a,\Re\mb b\right\rangle \right\Vert _{\psi_{1}}\\
 & +\left\Vert \left\langle \Re\mb a,\mb{\phi}\right\rangle \left\langle \mb{\phi},\Im\mb b\right\rangle -\left\langle \Re\mb a,\Im\mb b\right\rangle \right\Vert _{\psi_{1}}\\
 & +\left\Vert \left\langle \Im\mb a,\mb{\phi}\right\rangle \left\langle \mb{\phi},\Re\mb b\right\rangle -\left\langle \Im\mb a,\Re\mb b\right\rangle \right\Vert _{\psi_{1}}\\
 & +\left\Vert \left\langle \Im\mb a,\mb{\phi}\right\rangle \left\langle \mb{\phi},\Im\mb b\right\rangle -\left\langle \Im\mb a,\Im\mb b\right\rangle \right\Vert _{\psi_{1}}\\
 & \leq C\left(\left\Vert \Re\mb a\right\Vert _{2}+\left\Vert \Im\mb a\right\Vert _{2}\right)\left(\left\Vert \Re\mb b\right\Vert _{2}+\left\Vert \Im\mb b\right\Vert _{2}\right)\\
 & \leq2C\left\Vert \mb a\right\Vert _{2}\left\Vert \mb b\right\Vert _{2},
\end{align*}
 and 

\begin{align*}
\left\Vert \left|\left\langle \mb a,\mb{\phi}\right\rangle \right|^{2}-\left\Vert \mb a\right\Vert _{2}^{2}\right\Vert _{\psi_{1}} & \leq\left\Vert \left|\left\langle \Re\mb a,\mb{\phi}\right\rangle \right|^{2}-\left\Vert \Re\mb a\right\Vert _{2}^{2}\right\Vert _{\psi_{1}}+\left\Vert \left|\left\langle \Im\mb a,\mb{\phi}\right\rangle \right|^{2}-\left\Vert \Im\mb a\right\Vert _{2}^{2}\right\Vert _{\psi_{1}}\\
 & \leq8\left(\left\Vert \Re\mb a\right\Vert _{2}^{2}+\left\Vert \Im\mb a\right\Vert _{2}^{2}\right)=8\left\Vert \mb a\right\Vert _{2}^{2},
\end{align*}
which completes the proof.\end{proof}
\begin{lem}
\label{lem:RademacherSubExp2}For a Rademacher vector $\mb{\phi}$
with iid entries and any given vector $\mb a$, we have 
\begin{align*}
\left\Vert \left|\left\langle \mb a,\mb{\phi}\right\rangle \right|\right\Vert _{\psi_{1}} & \leq8\left\Vert \mb a\right\Vert _{2}.
\end{align*}
\end{lem}
\begin{proof}
We first treat the case of real vector $\mb a$ and then obtain the
general case from the real case. Using the Hoeffding's inequality
we can write 
\begin{align*}
\mbb E\left[e^{\frac{\left|\left\langle \mb a,\mb{\phi}\right\rangle \right|}{u}}\right] & =1+\int_{0}^{\infty}\mbb P\left(e^{\frac{\left|\left\langle \mb a,\mb{\phi}\right\rangle \right|}{u}}>e^{t}\right)e^{t}\mr dt\\
 & =1+\int_{0}^{\infty}\mbb P\left(\left|\left\langle \mb a,\mb{\phi}\right\rangle \right|>tu\right)e^{t}\mr dt\\
 & \leq1+2\int_{0}^{\infty}e^{t-\frac{t^{2}u^{2}}{2\left\Vert \mb a\right\Vert _{2}^{2}}}\mr dt\\
 & =1+2e^{\frac{\left\Vert \mb a\right\Vert _{2}^{2}}{2u^{2}}}\int_{0}^{\infty}e^{-\frac{1}{2}\left(\frac{tu}{\left\Vert \mb a\right\Vert _{2}}-\frac{\left\Vert \mb a\right\Vert _{2}}{u}\right)^{2}}\mr dt\\
 & \leq1+2\sqrt{2\pi}\frac{\left\Vert \mb a\right\Vert _{2}}{u}e^{\frac{\left\Vert \mb a\right\Vert _{2}^{2}}{2u^{2}}}.
\end{align*}
 In particular, at $u=4\sqrt{2}\left\Vert \mb a\right\Vert _{2}$
we have 
\begin{align*}
\mbb E\left[e^{\frac{\left|\left\langle \mb a,\mb{\phi}\right\rangle \right|}{4\sqrt{2}\left\Vert \mb a\right\Vert _{2}}}\right] & \leq1+\frac{\sqrt{\pi}}{2}e^{\frac{1}{64}}<2,
\end{align*}
which implies that $\left\Vert \left|\left\langle \mb a,\mb{\phi}\right\rangle \right|\right\Vert _{\psi_{1}}\leq4\sqrt{2}\left\Vert \mb a\right\Vert _{2}$. 

To obtain the complex version of the inequalities we can simply apply
the latter inequality to the real and imaginary parts of $\mb a$.
Then, we can write 
\begin{align*}
\left\Vert \left|\left\langle \mb a,\mb{\phi}\right\rangle \right|\right\Vert _{\psi_{1}} & \leq\left\Vert \left|\left\langle \Re\mb a,\mb{\phi}\right\rangle \right|\right\Vert _{\psi_{1}}+\left\Vert \left|\left\langle \Im\mb a,\mb{\phi}\right\rangle \right|\right\Vert _{\psi_{1}}\\
 & \leq4\sqrt{2}\left(\left\Vert \Re\mb a\right\Vert _{2}+\left\Vert \Im\mb a\right\Vert _{2}\right)\\
 & \leq8\left\Vert \mb a\right\Vert _{2},
\end{align*}
 where the first inequality is the triangle inequality, the second
inequality follows from the real version shown above, and the third
inequality is a simple application of the Cauchy-Schwarz inequality.
\end{proof}

\bibliographystyle{siam}
\bibliography{references}

\end{document}